\documentclass{article}
\usepackage[utf8]{inputenc}
\usepackage{placeins}
\usepackage{mystyle}
\usepackage{tocbibind}
\usepackage{hyperref}
\hypersetup{
  colorlinks=true,allcolors=blue!60!gray,
  pdftitle={Improved Online Load Balancing in the Two-Norm},
  pdfauthor={Sander Borst and Danish Kashaev}
}
\usepackage{hypcap}
\usepackage[capitalise]{cleveref}

\usepackage{mathtools}
\newcommand{\vb}{\nu}
\newcommand{\Exp}{\mathbb{E}}
\newcommand{\Var}{\operatorname{Var}}

\newcommand{\vh}{\hat{\nu}}
\newcommand{\J}{J} 
\newcommand{\M}{M} 
\newcommand{\x}{x}
\newcommand{\q}{q}
\usepackage[letterpaper,margin=1in]{geometry}
\newcommand{\maclambda}{\varepsilon}
\newcommand{\maceta}{\tilde{\varepsilon}}
\newcommand{\macpsi}{\lambda}
\newcommand{\tauone}{\tau}

\title{Improved Online Load Balancing in the Two-Norm\footnote{This paper will appear in the Proceedings of the 38th ACM-SIAM Symposium on Discrete Algorithms (SODA 2027). }}
\author{Sander Borst\footnote{\texttt{s.j.borst@outlook.com}, Max Planck Institute for Informatics, Saarbr\"ucken} \and Danish Kashaev \footnote{\texttt{kashaev.danish@gmail.com}, Centrum Wiskunde \& Informatica, Amsterdam}}
\newcommand{\valA}{1.0326}
\newcommand{\valB}{1.6208}
\newcommand{\valGamma}{\frac{1}{4.9843}}
\newcommand{\valLamb}{0.00253}
\newcommand{\valDelta}{0.02602}
\newcommand{\valEta}{0.00469}
\newcommand{\valPsi}{0.02753}
\newcommand{\valTheta}{0.0535}
\newcommand{\valTau}{0.14039}

\newcommand{\valCompRatio}{4.9843}
\newcommand{\valCompRatioApprox}{4.98}
\newcommand{\valGammaApprox}{0.2006}
\date{}

\begin{document}
\maketitle

\begin{abstract}
We study the online load balancing problem on unrelated machines, with the objective of minimizing the square of the $\ell_2$ norm of the loads on the machines. The greedy algorithm of Awerbuch et al. (STOC'95) is optimal for deterministic algorithms and achieves a competitive ratio of $3 + 2 \sqrt{2} \approx 5.828$, and an improved $5$-competitive randomized algorithm based on independent rounding has been shown by Caragiannis (SODA'08). In this work, we present the first algorithm breaking the barrier of $5$ on the competitive ratio, achieving a bound of $5 - c$ for some $c \geq 0.0157$. To obtain this result, we use a new primal-dual framework to analyze this problem based on a natural semidefinite programming relaxation, together with an online implementation of a correlated randomized rounding procedure of Im and Shadloo (SODA'20). This novel primal-dual framework also yields new, simple and unified proofs of the competitive ratio of the $(3 + 2 \sqrt{2})$-competitive greedy algorithm, the $5$-competitive randomized independent rounding algorithm, and that of a new $4$-competitive optimal fractional algorithm. We also provide lower bounds showing that the previous best randomized algorithm is optimal among independent rounding algorithms, that our new fractional algorithm is optimal, and that a simple greedy algorithm is optimal for the closely related online scheduling problem $R || \sum w_j C_j$.
\end{abstract}

\section{Introduction}
Scheduling a set of jobs over a set of machines in a balanced way is a fundamental problem in computer science, both in theory and practice. In many scenarios, decisions must be made without complete knowledge of future events. This motivates the study of \emph{online} scheduling, where jobs arrive over time and must be assigned to machines irrevocably upon arrival, without knowledge of future job characteristics.

In this work, we study the online load balancing problem on unrelated machines, defined as follows. Given is a set of machines $\M$ and a set of jobs $\J$ which arrives online in an adversarial order. When a job $j \in J$ arrives, it reveals a feasible subset of machines it can be assigned to with unrelated weights $w_{ij} \geq 0$ (or processing times) for scheduling the job on each of these machines $i \in \M$. An online algorithm then needs to irrevocably assign each arriving job to a machine. The goal is to minimize the sum of squares (or the square of the $\ell_2$ norm) of the \emph{loads} on the machines, where the load of a machine is defined as the total amount of weight assigned to it. This is a natural objective function which penalizes unbalanced solutions, favoring schedules where the workload on each machine is split fairly. Moreover, together with the \emph{makespan} (or \emph{maximum load}) objective, it is arguably the simplest non-trivial objective function for this problem, as the $\ell_1$ norm can be solved optimally online by a simple greedy algorithm. The quality of an algorithm is measured by the \emph{competitive ratio}, which is defined as the worst-case ratio, over all instances, between the cost of an online algorithm and that of the offline optimal solution.

This problem was first studied by Awerbuch et al., who showed that the greedy algorithm achieves a competitive ratio of $3 + 2\sqrt{2} \approx 5.828$   \cite{awerbuch1995load}.
This bound is in fact \emph{optimal} among deterministic algorithms. 
It took over a decade before a better randomized algorithm was presented by Caragiannis, which achieves a competitive ratio of $5$ \cite{caragiannis2008better}.
Nevertheless, understanding the full power of randomization in this setting remains a challenging open question. 
In this paper, we make further progress on this problem and introduce the first algorithm with a competitive ratio better than $5$, with the bound obtained being $5 - c$ for some $c \geq 0.0157$. 

The $5$-competitive algorithm uses an understanding of the first  and second moments of the loads on the machines in order to generate a fractional assignment at each online step using a \emph{water-filling} algorithm, which is then randomly rounded by assigning each job \emph{independently} to a machine according to the fractional probabilities. We in fact show a matching lower bound for such randomized \emph{independent} rounding algorithms. This means that more advanced online randomized rounding techniques are necessary in order to break the barrier of $5$, namely algorithms which satisfy some strong negative correlation properties. Moreover, one also needs to design an analysis technique in order to prove that such an algorithm would indeed attain an improved competitive ratio.

Even in the \emph{offline} setting, this turns out to be a difficult task, as illustrated by the approximability of the scheduling problem on unrelated machines to minimize the sum of weighted completion times, denoted as $R || \sum w_j C_j$. This problem is closely related to the load balancing problem we study, as both are assignment problems with a quadratic objective function. Similarly to the mentioned barrier of $5$, this offline scheduling problem admits a lower bound on the approximation ratio of $3/2$ against independent rounding algorithms, and the first algorithm breaking this barrier by a very small constant (at most $10^{-6}$) was a breakthrough result by \cite{bansal2016lift}, which introduced a correlated rounding algorithm with the required negative correlation properties in the offline setting.

Since then, there have been several works designing different dependent rounding algorithms with strong negative correlation properties \cite{im2020weighted, im2023improved, baveja2024approximating, harris2024dependent, naor2025online, harris2026dirichlet}. While the motivation for most of these works was to further improve the approximation ratio of $R || \sum w_j C_j$, such rounding algorithms have become an interesting object of study in their own right, as they can be applied more abstractly to assignment problems defined on a bipartite graph, where the goal is to pick exactly one edge for each node on one side of the bipartition. It turns out that the online load balancing problem we consider is a very natural problem for which such dependent rounding algorithms can be applied, and we believe that studying the limits of the competitive ratios achievable for this problem is closely tied to understanding the power of such dependent rounding algorithms. This work is focused on building a primal-dual framework allowing to analyze such dependent rounding algorithms for this problem, and showing that a competitive ratio bound strictly better than $5$ is indeed possible.

\subsection{Our contributions}
Our first contribution is a new primal-dual framework for analyzing online algorithms for this load balancing problem. We show that this framework yields simple and unified proofs for the previous best deterministic and randomized algorithms, as well as for a new \emph{optimal} algorithm for the fractional relaxation of the problem.  The framework is based on the dual of a natural semidefinite programming relaxation, and builds on a dual-fitting technique introduced in \cite{kashaev2026selfish}.  Although the primal-dual method is a well-known tool in the design and analysis of online algorithms, it has to the best of our knowledge rarely been used with semidefinite programming relaxations.

\begin{theorem}
A new primal-dual framework based on a natural semidefinite programming relaxation allows for a unified analysis of
\begin{itemize}
\item the optimal deterministic $(3 + 2 \sqrt{2})$-competitive greedy algorithm of \cite{awerbuch1995load},
\item the $5$-competitive randomized algorithm based on independent rounding of \cite{caragiannis2008better}, and
\item a new optimal $4$-competitive fractional algorithm.
\end{itemize}
\end{theorem}

We then use this primal-dual framework to design a
randomized algorithm that achieves a competitive ratio strictly better than $5$, which is our main contribution.

\begin{theorem}
There exists a $(5 - c)$-competitive randomized primal-dual algorithm for the online load balancing problem on unrelated machines, where $c \geq 0.0157$.
\end{theorem}

While the numerical improvement is relatively small,
we believe this algorithm represents an important conceptual step. As opposed to the previous best randomized algorithm, which is based on \emph{independent} rounding, our algorithm uses a \emph{dependent} rounding procedure that introduces strong negative correlation properties between the jobs. This procedure requires a grouping of the jobs on each machine, and stronger negative correlation properties are ensured for the jobs inside of each group. An important contribution of our work is to design an online grouping for this problem, by identifying a class of jobs for which strong negative correlation properties are necessary.

We also show that such a correlated rounding procedure is necessary to break the barrier of $5$.
Namely, we present a new lower bound instance for this problem for which no fractional algorithm can be better than $4$-competitive and no randomized algorithm making \emph{independent} random choices for assigning the jobs can be better than $5$-competitive. This shows that a competitive ratio of $5$ is an important barrier for this problem, analogously to the barrier of $3/2$ known for $R || \sum w_j C_j$ against offline independent rounding algorithms \cite{bansal2016lift}. Note that our lower bound of $4$ also applies to arbitrary randomized algorithms, since any randomized algorithm induces a fractional algorithm with lower or equal cost for this problem, due to the convexity of the objective function.

\begin{theorem}
There exists an adversarial instance for the online load balancing problem for which
\begin{itemize}
    \item no fractional or randomized algorithm can be better than $4$-competitive, and
    \item no randomized algorithm based on independent rounding can be better than $5$-competitive.
\end{itemize}
\end{theorem}

As a secondary result, we also consider the scheduling problem $R || \sum w_j C_j$ in the online setting. It is known that a simple deterministic greedy algorithm is $4$-competitive \cite{gupta2020greed} for this problem. We provide a lower bound instance showing that no fractional or randomized algorithm can be better than $4$-competitive, obtained through a reduction from our hard instance for the online load balancing problem. 

\begin{theorem}
There exists an adversarial instance for the online problem $R || \sum w_j C_j$ showing that no fractional or randomized algorithm can be better than $4$-competitive.
\end{theorem}

\subsection{Our techniques}
\paragraph*{Primal-dual analysis.}
To analyze the different algorithms considered for this problem, we use a natural semidefinite programming relaxation, which can essentially be obtained through the first round of the Lasserre hierarchy \cite{lasserre2001global}. The analysis technique then consists of constructing a feasible solution to the dual of that relaxation whose value is at least a factor $\rho \in [0,1]$ times the cost incurred by the considered online algorithm, which would imply a competitive ratio bound of $1 / \rho$. This technique was introduced in \cite{kashaev2026selfish} to mostly analyze inefficiency of equilibria in a game-theoretic setting. At a high-level, the online algorithms considered in this paper all satisfy an \emph{equilibrium} condition after each job arrival. The inequalities derived from such an equilibrium condition are similar in structure to the first set of constraints of our SDP relaxation. The dual fitting is then chosen in a way that ensures dual feasibility is satisfied due to the equilibrium inequalities. For an accessible way to understand this approach in detail, we invite the reader to read the beginning of Section \ref{sec_greedy_alg}.

To illustrate the power of this SDP approach on this problem, we first use it to give simple and unified analyses of the $(3 + 2 \sqrt{2})$-competitive greedy algorithm of \cite{awerbuch1995load} and the $5$-competitive randomized algorithm based on independent rounding of \cite{caragiannis2008better}. Our analysis of the latter algorithm then provides new insights into the problem. In particular, it becomes apparent that for every machine, the analysis is only tight for a specific family of jobs, which we refer to as the \emph{hard} jobs for that machine. These jobs are inherently difficult to schedule using independent rounding and form the bottleneck for improving the competitive ratio. This observation is key in the design of our new randomized algorithm. We give more concrete details on how the hard jobs are defined at the beginning of \cref{sec_new_rand_algo}.

\paragraph*{Online correlated rounding.}
Given the lower bound of $5$ that we present against independent rounding algorithms, it is thus necessary to use a more sophisticated online rounding algorithm to improve the competitive ratio.
To see the limitations of independent rounding, denote by $X_{ij} \in \{0,1\}$ the indicator random variable of whether $j \in \J$ is assigned to $i \in \M$ by a given algorithm, denote its expectation as $x_{ij}:= \Exp[X_{ij}]$ and denote by $L_i := \sum_{j \in \J} w_{ij} X_{ij}$ the \emph{load} of machine $i$. The expected objective value of the algorithm is:
\begin{align*}
    \sum_{i \in \M} \Exp\left[ L_i ^2 \right]=   \sum_{i \in M} \Exp\left[ L_i \right]^2 +  \sum_{i \in \M} \Var\left[ L_i \right].
\end{align*}
In other words, this is equal to the objective value of the fractional solution in addition to the sum of the variances of the loads on the machines. Under \emph{independent rounding}, an easy computation shows that
\[\sum_{i \in \M} \Var\left[ L_i \right] = \sum_{i \in \M} \sum_{j \in \J} w_{ij}^2 \: x_{ij} (1 - x_{ij}).\]
In the setting where the weight of a job is the same on every feasible machine, meaning that $w_{ij} = w_{j}$, and if $x_{ij} \approx 0$ for every pair $i,j$, the above expression can become $\approx \sum_{j \in \J}w_j^2$. This term can become close to the value of an optimal solution, which intuitively explains why the competitive ratio of the best randomized independent rounding algorithm is $5$, whereas the best possible algorithm for the fractional relaxation is $4$-competitive.

Hence, it seems like a natural approach to use some form of correlated rounding to try to decrease the sum of the variances. 
However, it is not obvious how to correlate the jobs in an online way.  
Moreover, the extent to which the variance term on a machine can be reduced also depends on the distribution of the weights of the jobs. 
For example, if a large fraction of the load on a machine comes from a single job, its variance cannot significantly  be decreased by introducing negative correlation.

A key insight of our work is that by negatively correlating the \emph{hard} jobs that we identified using the primal-dual analysis of the $5$-competitive algorithm of \cite{caragiannis2008better}, we can significantly reduce the variance caused by these jobs.
This allows us to reduce the part of the objective function caused by these hard jobs by a constant factor.
The part of the objective function caused by all other jobs, called \emph{easy} jobs, can be reduced by appropriately optimizing the fractional solution and the corresponding dual SDP variables. 

To introduce negative correlation between the hard jobs in an online manner, we build upon a correlated rounding scheme of \cite{im2020weighted}, which was originally developed for the offline problem $R || \sum w_j C_j$. We show that their rounding scheme can be implemented online, while preserving the desired negative correlation properties. 
This scheme allows us to group the hard jobs together per machine and negatively correlate the jobs within each group, while maintaining the marginal probabilities. Note that the groups can differ per machine, which is crucial as the set of hard jobs may differ per machine as well. 

We then analyze the competitiveness of the algorithm using the primal-dual framework mentioned above. The negative correlation properties of the rounding scheme allow us to bound the variance term more tightly than in the independent rounding case, leading to an improved competitive ratio. The dual fitting becomes much more involved than for the previous algorithms, and several new ideas are required in order to link the dual objective value to the expected cost of the correlated rounding algorithm while simultaneously keeping dual feasibility at all times. The dual variables now need to be carefully updated at each online step by also taking into account how our grouping algorithm evolves over time for each machine. For more intuition on the techniques used to update the dual variables, we invite the reader to read the beginning of \cref{sec_analysis} after gaining familiarity with the SDP relaxation.

\subsection{Related work}

\paragraph*{Online scheduling.} In the online unrelated setting, \cite{awerbuch1995load} show that the greedy algorithm is $3 + 2\sqrt{2} \approx 5.828$-competitive for minimizing the square of the $\ell_2$ norm. 
This was shown to be a tight bound in \cite{caragiannis2011tight} even in the restricted identical machines setting, and this bound is even the best possible among deterministic algorithms \cite{caragiannis2008better}. 
Improving this bound was an open question until a $5$-competitive randomized algorithm was shown by \cite{caragiannis2008better}. 
In fact, this approach improved the best known bounds for the more general $\ell_p$ norm for many values of $p\geq 2$. Different variants and generalizations of this problem have been studied in \cite{suri2004selfish, caragiannis2011tight, molinaro2017online, avidor2001ancient, im2018online, berman2000line, im2015tight, azar2016online}.

Many other models for online scheduling of jobs on machines have been studied. One natural model consists of jobs having a release date and arriving online at that point in time, where the objective is to minimize an objective function depending on the weighted flow time of jobs. In this model, strong lower bounds are known, even for preemptive algorithms \cite{kellerer1996approximability, garg2007minimizing, chekuri2001algorithms}. 
Given these lower bounds, such scheduling problems have been considered in the speed augmentation model, where each machine is allowed to run at an $\varepsilon$-fraction faster speed than the offline optimum \cite{chadha2009competitive, im2011online, kalyanasundaram2000speed, bansal2003server}. 
Dual fitting approaches on LPs and convex programs for different scheduling problems have been developed in \cite{anand2012resource, jager2025power, im2014selfishmigrate, garg2019non, im2017competitive, gupta2012online, gupta2020greed, jager2023improved}.

\paragraph*{Offline scheduling.}
In the offline setting, the unrelated quadratic load balancing problem admits an easy $2$-approximation by rounding a convex program \cite{azar2005convex}. The best known approximation algorithm obtains a bound of $4/3$ and is given in \cite{im2023improved} by using a time-indexed LP and the Shmoys-Tardos rounding algorithm \cite{shmoys1993approximation}. For the more general $\ell_p$ norm with $p < \infty$, \cite{awerbuch1995load} show how to get a $\Theta(p)$ approximation. In a breakthrough, \cite{azar2005convex} improved this to a $2$-approximation by using (again) the Shmoys-Tardos rounding algorithm. Further improvements have been made in \cite{kumar2009unified} by a new dependent rounding approach. The $\ell_\infty$ norm corresponds to the makespan minimization problem. It is known to be NP-hard to approximate within a factor of 1.5, and a $2$-approximation is given in the classic result of \cite{shmoys1993approximation}. 

Another closely related scheduling problem in the offline setting is $R || \sum w_j C_j$, where the goal is to minimize the weighted sum of completion times of the jobs, as that is also an assignment problem with a quadratic objective function. 
This problem is APX-hard \cite{hoogeveen1998non} and admits a $3/2$-approximation based on independent rounding \cite{skutella2001convex, sethuraman1999optimal}, a bound which is the best possible among algorithms making independent random choices for each job \cite{bansal2016lift}. 
A breakthrough result of \cite{bansal2016lift} showed how to get a $3/2 - c$ approximation for $c \leq 10^{-6}$ by rounding a novel semidefinite programming relaxation through a dependent rounding procedure introducing the desired negative correlation properties. Subsequent improvements on the approximation ratio have been made in \cite{im2020weighted, im2023improved, harris2024dependent, li2025approximating}. The original rounding procedure of \cite{bansal2016lift} is a variant of \emph{pipage rounding}, a technique which has been used for other non-scheduling problems \cite{ageev2004pipage, chekuri2010dependent, gandhi2006dependent, srinivasan2001distributions}. However, this 
technique seems difficult to implement in an online manner.
A different rounding scheme, developed by \cite{im2020weighted}, uses a variant of \emph{fair contention resolution}, which was originally developed for allocation problems \cite{feige_approximation_2006}. We show that this rounding scheme can be implemented online. This procedure is a key ingredient of our algorithm.
 
\subsection{Outline of the paper} Section \ref{sec_prelimin} is devoted to preliminaries introducing the online load balancing problem and the semidefinite programming relaxation used. The analysis of the $(3 + 2 \sqrt{2})$-competitive greedy algorithm of \cite{awerbuch1995load}, of the $5$-competitive randomized algorithm of \cite{caragiannis2008better} and of our $4$-competitive optimal fractional algorithm based on this new primal-dual approach are presented in Section \ref{sec_unified_SDP_analyses}. The new improved $\approx \valCompRatioApprox$-competitive randomized algorithm is presented in Section \ref{sec_new_algo} and its analysis is shown in Section \ref{sec_analysis}. Matching lower bounds against fractional and randomized independent rounding algorithms are presented in Section \ref{sec_lower_bounds}.

\section{Preliminaries}
\label{sec_prelimin}
\paragraph*{The online load balancing problem.} A set of machines $\M$ is given. A set of jobs $\J$ arrives online in an adversarial order. Each time a job $j \in \J$ arrives, it reveals a subset $\mathcal{S}_j \subseteq \M$ of machines it can be assigned to with unrelated weights $w_{ij} \geq 0$ associated with every $i \in \mathcal{S}_j$. An online deterministic integral algorithm needs to irrevocably pick a machine $i \in \mathcal{S}_j$ to assign an arrived job to. We denote by $x_{ij} \in \{0,1\}$ the indicator variable whether an algorithm assigns $j$ to $i$. For $i \notin \mathcal{S}_j$, we define $x_{ij} = 0$. The load of a machine $i \in \M$ is the total amount of weight assigned to it and is denoted as
\[L_i(x) = \sum_{j \in \J} w_{ij} \: x_{ij}.\]
Since jobs arrive online, we order them as $\J = \{1, \dots n\}$, where job $j$ arrives before job $k$ if $j < k$. For each machine $i \in \M$, we denote the load over time as 
\[L_i ^{(j)}(x) = \sum_{k \leq j} w_{ik} \: x_{ik} \qquad \forall j \in \J.\]
In words, this is the load of a machine after job $j$ has been assigned. Observe that the final load is $L_i(x) = L_i^{(n)}(x)$. We define $L_i^{(0)} := 0$ for convenience. The goal of the problem is to minimize the following objective function, which is the sum of squares of the loads:
\begin{equation}
\label{eq_objective_function}
C(x)= \sum_{i \in \M} L_i(x)^2.
\end{equation}

\paragraph*{Fractional and randomized algorithms.}
Whenever a job $j \in \J$ arrives, a deterministic \emph{fractional} algorithm is allowed to pick $x_{ij} \in [0,1]$ for every $i \in \mathcal{S}_j$ while satisfying the constraint $\sum_{i \in \M}x_{ij} = 1$. A \emph{randomized integral} algorithm may pick a machine $i \in \M$ to assign job $j \in \J$ to at random, which induces indicator random variables that we denote by $X_{ij} \in \{0,1\}$. The expected cost of a randomized algorithm is denoted as 
\[\sum_{i \in \M}\mathbb{E}\left[L_i(X)^2\right] \qquad \text{where} \qquad L_i(X) := \sum_{j \in \J} w_{ij} \: X_{ij}.\]
Note that a randomized integral algorithm induces a deterministic fractional algorithm with lower or equal cost, by defining $x_{ij} := \mathbb{E}[X_{ij}]$ and using Jensen's inequality, implying $L_i(x)^2 = \mathbb{E}[L_i(X)]^2 \leq \mathbb{E}[L_i(X)^2]$. For randomized algorithms, we will often denote $L_i = L_i(X)$ for simplicity of notation.

\paragraph*{The semidefinite programming relaxation.}
Observe that the objective function $C(x)$ can be written as 
\begin{equation}
\label{eq_soc_cost_quadratic}
C(x) = \sum_{i \in \M} \sum_{j,k \in \J} w_{ij} \: w_{ik} \:x_{ij} x_{ik}.
\end{equation}
An exact binary quadratic program to compute the offline optimal solution $x^*$ is then given by:
\begin{align*}
    \min \:  &C(x) \\
    &\sum_{i \in \mathcal{S}_j} x_{ij} = 1 \qquad \forall j \in \J \\
    &x_{ij} \in \{0,1\} \qquad \forall j \in \J, \forall i \in \mathcal{S}_j.
\end{align*}

We now consider a semidefinite convex relaxation of the above quadratic program, which can essentially be obtained through the Lasserre hierarchy \cite{lasserre2001global}. The variable of the program is a positive semidefinite matrix $X$ of dimension $1 + \sum_{j \in \J}|\mathcal{S}_j|$, which has one row/column corresponding to each $x_{ij}$, in addition to one extra row/column that we index by $0$. 

\begin{align}
\label{eq_primal_sdp}
    \min \sum_{i \in \M} \sum_{j,k \in \J} w_{ij} \: w_{ik} \: X_{\{ij, ik\}} &\qquad  \qquad \\
    \sum_{i \in \mathcal{S}_j} X_{\{ij,\: ij\}} &= 1 \hspace{4cm} \forall j \in \J \nonumber \\
    X_{\{0,0\}} & = 1 \nonumber\\
    X_{\{0, \: ij\}} &= X_{\{ij, \: ij\}} \hspace{3cm} \forall j \in \J, i \in \mathcal{S}_j \nonumber\\
    X_{\{ij, \: i'k\}} &\geq 0 \hspace{4cm} \forall {(i,j), (i',k)} \text{ with } j,k \in \J. \nonumber\\
    X & \succeq 0 \nonumber
\end{align}
To see that this is in fact a relaxation to the previous quadratic program computing the offline optimum, note that for any binary feasible assignment $x$, the rank-one matrix $X = (1,x) (1, x)^T$ is a feasible solution to the SDP with the same objective value, since $x_{ij}^2 = x_{ij}$ for $x_{ij} \in \{0,1\}$, leading to $X_{\{ij, \: ij\}} = x_{ij}^2 = x_{ij} = X_{\{0, \: ij\}}$. The dual to this relaxation, written in vector form, is the following. The derivation of the dual is shown in Appendix \ref{sec_take_dual}. We call this relaxation \emph{(SDP-LB)}.
\vspace{0.1cm}
\begin{align}
\label{dual_sdp}
    \max \sum_{j \in \J} y_j  - &\frac{1}{2}\Vert \vb \Vert ^ 2 \\
    y_j &\leq w_{ij}^2 - \frac{1}{2}\Vert v_{ij} \Vert ^ 2 + \: \langle \vb, v_{ij} \rangle \qquad \qquad \forall j \in \J, i \in \mathcal{S}_j \label{SDPcstr} \\
    \langle v_{ij}, v_{i'k} \rangle &\leq 2 \; w_{ij} \: w_{i'k} \: \mathds{1}_{\{i = i'\}} \hspace{2.3cm} \forall (i,j) \neq (i',k) \text{ with } j,k \in \J. \label{SDPcstr2}
\end{align}
The variables of this program are real-valued $y_j \in \mathbb{R}$ for every $j \in \J$, as well as vectors $\vb \in \mathbb{R}^M$ and $v_{ij} \in \mathbb{R}^M$ for every $j \in \J, i \in \mathcal{S}_j$. By weak duality, every dual feasible  solution provides a lower bound on the optimal offline value of the original problem.

\section{A unified online primal-dual approach}
\label{sec_unified_SDP_analyses}
\subsection{The greedy algorithm}
\label{sec_greedy_alg}
{\renewcommand{\baselinestretch}{1.3}\selectfont{
\begin{algorithm}
    \caption{\Call{Greedy}{}}
    \label{greedy_algo_LB}
    \begin{algorithmic}
     \When{$j \in \J$ arrives}
     \State Set $x_{ij} = 1$ for $i \in \mathcal{S}_j$ giving the minimal increase in the global objective function
     \EndWhen
     \State \Return $x$
 \end{algorithmic}
 \end{algorithm}
 }}
We first consider the \Call{Greedy}{} algorithm. Whenever a job $j \in \J$ arrives, \Call{Greedy}{} picks $i \in \mathcal{S}_j$ (i.e. sets $x_{ij} = 1$) which gives the least increase in the global objective function. Let us denote the cost incurred by the algorithm after $j \in J$ has been assigned by 
\[C^{(j)}(x) = \sum_{e \in M} \left(L_e^{(j)}(x)\right)^2.\]
The key property of the greedy algorithm is the following lemma.

\begin{lemma}
\label{lemma_greedy_ineq}
For any instance and any solution $(x_{ij})_{j \in \J, i \in \mathcal{S}_j}$ constructed by \Call{Greedy}{}, the following inequalities are satisfied:
\[C^{(j)}(x) - C^{(j-1)}(x)  \leq w_{ij}^2 + 2 \: L_i^{(j-1)}(x)\: w_{ij} \qquad \forall j \in \J, \forall i \in \mathcal{S}_j.\]
\end{lemma}
\begin{proof}
Let us fix an arbitrary $j \in \J$. By definition of the greedy algorithm, whenever $j \in \J$ arrives, one has the following inequality:
\[C^{(j)}(x) \leq \sum_{e \in \M} \Big(L_e^{(j-1)}(x) + w_{ej} \mathds{1}_{\{e = i\}}\Big)^2 \qquad \forall i \in \mathcal{S}_j.\]
Expanding out the right-hand side and rearranging terms gives:
\[C^{(j)}(x) - C^{(j-1)}(x) \leq \sum_{e \in \M} \Big(w_{ej}^2 \mathds{1}_{\{e = i\}} + 2 \: L_e^{(j-1)}(x) \: w_{ej} \mathds{1}_{\{e = i\}} \Big) \quad \forall i \in \mathcal{S}_j. \qedhere\]
\end{proof}

We will now analyze the competitive ratio of \Call{Greedy}{} through a dual fitting argument on our semidefinite programming relaxation. We now describe a high-level view of our approach, as it will also help build intuition for the analyses of the following randomized algorithms. The key idea is to make the first set of SDP constraints \eqref{SDPcstr} satisfied due to the equilibrium inequalities of the algorithm stated in Lemma \ref{lemma_greedy_ineq}. Hence, it seems natural to fit 
\[y_j = \gamma \left[C^{(j)}(x) - C^{(j-1)}(x) \right], \quad v_{ij}(e) = \alpha \: w_{ij} \mathds{1}_{\{e = i\}}, \quad \vb(e) = \beta L_e(x)\]
for some constants $\alpha, \beta, \gamma \geq 0$ that are to be determined. Observe that we need $\alpha^2 \leq 2$, since the second set of SDP constraints \eqref{SDPcstr2} will now be satisfied if that holds. Under our fitting, \eqref{SDPcstr} becomes
\begin{align*}
y_j \leq w_{ij}^2 - \frac{1}{2}\Vert v_{ij} \Vert ^ 2 &+ \: \langle \vb, v_{ij}  \rangle \\
\iff &\gamma \left[C^{(j)}(x) - C^{(j-1)}(x) \right] \leq \left(1 - \frac{\alpha^2}{2}\right)w_{ij}^2 + \alpha \beta \: L_i(x) \: w_{ij}.
\end{align*}
Hence, if the constants satisfy $\gamma = 1 - \alpha^2/2 = \alpha \beta /2$, these constraints would be satisfied, due to Lemma \ref{lemma_greedy_ineq}. Finally, under our fitting, the dual objective function becomes 
\[\sum_{j \in \J} y_j  - \frac{1}{2}\Vert \vb \Vert ^ 2 = \gamma \: C(x) - \frac{\beta^2}{2} \sum_{e \in E} L_e(x)^2 = \left(\gamma - \frac{\beta^2}{2}\right) C(x)\]
where $C(x)$ is the final cost incurred by \Call{Greedy}{}. Since we want to maximize the dual objective in order to get the best possible competitive ratio, we would want to solve
\[\max_{\alpha, \beta, \gamma} \Big\{\gamma - \beta^2/2 \quad s.t. \quad \gamma = 1 - \alpha^2/2 = \alpha \beta /2, \quad \alpha^2 \leq 2, \quad \alpha \geq 0, \quad \beta \geq 0\Big\}\]
in order to find the right constants. These are now stated in the next lemma, and a formal proof of the competitive ratio of our algorithm is given in the following theorem.

\begin{lemma}
\label{lemma_constants_LB}
Let $\alpha, \beta \geq 0$ be defined as $\alpha = 2^{1/4}$ and $\beta = (2 - \alpha^2)/\alpha = (\sqrt{2}-1) \alpha$. The following properties hold:
\begin{itemize}
\item $1 - \alpha^2/2 = \alpha \beta/2$
\item $(\alpha \beta - \beta^2)/2 = 1/(3 + 2 \sqrt{2})$
\end{itemize}
\end{lemma}
\begin{proof}
The first property is immediate by the definition of $\beta$. The second property consists of simple computations whose proof is omitted.
\end{proof}
\begin{theorem}
For any instance of the online load balancing problem, and any solution $(x_{ij})_{j \in \J, i \in \mathcal{S}_j}$ obtained by \Call{Greedy}{}, there exists a feasible (SDP-LB) solution with objective value at least
\[\frac{1}{3 + 2 \sqrt{2}} \: \sum_{e \in \M} \: L_e(x)^2.\]
By weak duality, this implies that the competitive ratio of \Call{Greedy}{} is at most $3 + 2 \sqrt{2} \approx 5.828$.
\end{theorem}
\begin{remark}
This recovers the result given in \cite{awerbuch1995load}. Moreover, this bound is tight with a matching lower bound given in \cite{caragiannis2011tight} for restricted identical machines, meaning that $w_{ej} = w_j$ for every $j \in \J, i \in \mathcal{S}_j$. In fact, it even turns out that this algorithm is best possible among deterministic (integral) algorithms \cite{caragiannis2008better}. Moreover, our analysis also works in a more general congestion model, where jobs can be assigned to hyperedges of machines, see Appendix \ref{sec_hypergraph_gen}.
\end{remark}
\begin{proof}
The vectors of the SDP live in the space $\mathbb{R}^\M$. Let $\alpha, \beta \geq 0$ be defined as in Lemma \ref{lemma_constants_LB}. We now state the dual fitting:
\begin{itemize}
\item $\vb(e) := \beta \: L_e(x)$
\item $v_{ij}(e) := \alpha \: w_{ij} \: \mathds{1}_{\{e = i\}}$ \qquad \qquad \qquad  $\hspace{2cm} \forall j \in \J, i \in \mathcal{S}_j$
\item $y_j := \frac{\alpha \beta}{2} \left( C^{(j)}(x) - C^{(j-1)}(x) \right) \hspace{2.7cm} \forall j \in \J.$
\end{itemize}
The different inner products and norms of the vectors are the following.
\begin{align*}
\Vert \vb \Vert ^ 2 = \beta^2 \; \sum_{e \in \M} \: L_e(x)^2& \qquad
\Vert v_{ij} \Vert ^ 2 = \alpha^2 \;  w_{ij}^2 \\
\langle \vb, v_{ij} \rangle = \alpha \beta \;  w_{ij} \: L_i(x)& \qquad
\langle v_{ij}, v_{i'k} \rangle = \alpha^2 \; w_{ij} \: w_{i'k} \: \mathds{1}_{\{i = i'\}}.
\end{align*}
Let us now check feasibility of the solution. The second set of constraints is satisfied by the fourth computation above and the fact that $\alpha^2 \leq 2$. The first set of constraints turns out to be satisfied due to the inequalities valid for \Call{Greedy}{} stated in Lemma \ref{lemma_constants_LB}. Indeed, under the above fitting, for every $j \in \J, i \in \mathcal{S}_j$, the first set of SDP constraints reads:
\begin{align*}
y_j \leq w_{ij}^2 - \frac{1}{2}\Vert & v_{ij} \Vert ^ 2 + \: \langle \vb, v_{ij}  \rangle  \\
\iff &\frac{\alpha \beta}{2} \left( C^{(j)}(x) - C^{(j-1)}(x) \right) \leq \left(1 - \frac{\alpha^2}{2}\right)  \: w_{ij}^2 + \frac{\alpha \beta}{2} \;  2 \: w_{ij} \: L_i(x).
\end{align*}
By the first property of Lemma \ref{lemma_constants_LB}, the constant terms cancel out on both sides of the inequality. These inequalities are then clearly satisfied by Lemma \ref{lemma_greedy_ineq}, since $L_e^{(j-1)}(x) \leq L_e(x)$. The objective function becomes
\[\sum_{j \in \J} y_j - \frac{1}{2} \Vert \vb \Vert ^ 2 = \frac{\alpha \beta}{2} C(x) - \frac{\beta^2}{2} \; \sum_{e \in \M} \: L_e(x)^2 = \left(\frac{\alpha \beta}{2} - \frac{\beta^2}{2}\right) \:C(x) = \frac{1}{3 + 2 \sqrt{2}} \: C(x)\]
where the first equality follows since the sum of the $y$ variables is a telescoping sum and the last equality follows by the second property of Lemma \ref{lemma_constants_LB}.   
\end{proof}

\subsection{An improved randomized algorithm}\label{sec:randomized_algo}
\label{sec_ind_rounding_algo}

In this section, we show how to improve on the greedy algorithm using randomization to obtain a $5$-competitive algorithm. The algorithm we present is due to \cite{caragiannis2008better} and is called \Call{Balance}{}. We show how to analyze this algorithm using our dual SDP framework, yielding a simpler analysis. In Section \ref{sec_lower_bound_ind_rounding}, we present an adversarial instance showing that this bound is the best possible among randomized algorithms making \emph{independent} random choices for every job $j \in \J$.

Whenever $j \in \J$ arrives, we consider the following potential functions $f_{ij}:[0,1] \to \mathbb{R}$ for every $i \in \mathcal{S}_j$\footnote{Slightly different functions are actually used in \cite{caragiannis2008better}: $f_{ij}(t) := w_{ij}^2 + 4 \left(\: w_{ij} \: \mathbb{E}\left[L_i^{(j-1)}\right] + t \:w_{ij}^2\right)$. We use our definition as it makes the analysis slightly cleaner, but the same proof with minor changes would work in that case.}:
\begin{equation}
\label{f_function}
f_{ij}(t) := w_{ij}^2 + 4 \: w_{ij} \: \mathbb{E}\left[L_i^{(j-1)}\right] + t \:w_{ij}^2.
\end{equation}
The algorithm \Call{Balance}{} then defines a probability distribution $(x_{ij})_{i \in S_j}$ which ensures that
\[x_{ej} > 0 \implies f_{ej}(x_{ej}) \leq f_{ij}(x_{i j}) \qquad \forall i \in \mathcal{S}_j.\]
This can be achieved by continuously increasing the lowest potentials until sending a total fractional amount of one. Observe that this means that for every $e \in S_j$ with $x_{ej} > 0$, we get that $f_{ej}(x_{ej}) = \mu$ for some constant $\mu$, whereas $f_{ej}(x_{ej}) \geq \mu$ if $x_{ej} = 0$. In particular, we also get the following inequality:
\begin{equation}
\label{equilib_cond}
\sum_{e \in \M} x_{ej} \: f_{ej}(x_{ej}) \leq f_{ij}(x_{ij}) \qquad \forall i \in \mathcal{S}_j.
\end{equation}
Once the fractional assignment has been constructed, the algorithm independently samples a machine to assign the job to according to those probabilities. \Call{Balance}{} can be seen as a water-filling algorithm, combined with independent rounding.

{\renewcommand{\baselinestretch}{1.3}\selectfont{
\begin{algorithm}
    \caption{\Call{Balance}{}}
    \label{caragiannis_algo}
    \begin{algorithmic}
    \When{$j \in \J$ arrives}
    \vspace{0.1cm}
    \State Compute $(x_{ij})_{i \in \mathcal{S}_j}$ such that $\sum_{i \in M} x_{ij} = 1$ and $\sum_{e \in \M} x_{ej} \: f_{ej}(x_{ej}) \leq f_{ij}(x_{ij}) \qquad \forall i \in \mathcal{S}_j$
    \State Assign $j$ to $i \in \mathcal{S}_j$, i.e. set $X_{ij} = 1$ with probability $x_{ij}$
    \EndWhen
    \State \Return $X$
    \end{algorithmic}
 \end{algorithm}
 }}

Let us now analyze this algorithm, we first define:
\begin{align}
\delta_e^{(j)}(x) &:=  \mathbb{E}\left[L_e^{(j)}\right]^2 - \mathbb{E}\left[L_e^{(j-1)}\right]^2, \nonumber\\
\Delta_e^{(j)}(x) &:= \mathbb{E}\left[{L_e^{(j)}}^2\right] - \mathbb{E}\left[{L_e^{(j-1)}}^2\right] \label{eq_increments_moments}.
\end{align}
The first quantity is the difference, after $j$ has arrived, between the squares of the first moments of the load of a machine $e \in \M$, whereas the second one is the difference of the second moments. Let us now expand the definition of these terms. Observe that for every $e \in \M$, we have $L_e^{(j)} = L_e^{(j-1)} + X_{ej} \: w_{ej}$, meaning that the squares of the random loads vary as follows:
\[{L_e^{(j)}}^2 - {L_e^{(j-1)}}^2 = X_{ej} \: w_{ej}^2 + 2 \: L_e^{(j-1)} \:  X_{ej} \: w_{ej}\]
where we use the fact that $X_{ej}^2 = X_{ej}$, since it is a binary random variable. The term $\Delta_e^{(j)}(x)$ is simply the expectation of the above equation and can be written as follows:
\begin{align}
\label{eq_big_delta}
\Delta_e^{(j)}(x) = x_{ej} \: w_{ej}^2 + 2 \: \mathbb{E}[L_e^{(j-1)}] \:  x_{ej} \: w_{ej}
\end{align}
where we use that $\mathbb{E}[X_{ej}] = x_{ej}$ and the fact that \Call{Balance}{} makes \emph{independent} random choices for different jobs, meaning that the random variables $L_e^{(j-1)}$ and $X_{ej}$ are independent. In addition, we get:
\begin{align}
\label{eq_small_delta}
\delta_{e}^{(j)}(x) &= \left(\mathbb{E}[L_e^{(j-1)}] + x_{ej} \: w_{ej}\right)^2 - \left(\mathbb{E}[L_e^{(j-1)}]\right)^2 = x_{ej}^2 \: w_{ej}^2 + 2 \: \mathbb{E}[L_e^{(j-1)}] \:  x_{ej} \: w_{ej}.
\end{align}
Observe that the only difference between $\delta_{e}^{(j)}(x)$ and $\Delta_{e}^{(j)}(x)$ is the square for $x_{ej}$ in the first term. We can now relate the sum of these two quantities to the potential function \eqref{f_function}:
\begin{equation}
\label{eq_obj_ineq_caralgo}
\delta_e^{(j)}(x) + \Delta_e^{(j)}(x) = x_{ej} \left(w_{ej}^2 + 4 \: w_{ej} \: \mathbb{E}[L_e^{(j-1)}] + x_{ej} \: w_{ej}^2\right) = x_{ej} \: f_{ej}(x_{ej}).
\end{equation}
The following lemma states key inequalities satisfied by \Call{Balance}{} needed for the dual fitting.
\begin{lemma}
\label{lemma_balance_ineq_LB}
For any instance and any solution $(x_{ij})_{j \in \J, i \in \mathcal{S}_j}$ constructed by \Call{Balance}{}, the following inequalities are satisfied for every $j \in \J$:
\[\sum_{e \in \M} x_{ej}  \: f_{ej}(x_{ej}) \leq w_{ij}^2 + 4 \: w_{ij} \: \mathbb{E}\left[L_i\right] \qquad \forall i \in \mathcal{S}_j.\]
\end{lemma}
\begin{proof}
By the equilibrium condition \eqref{equilib_cond} of the algorithm, for every $i \in \mathcal{S}_j$, we get: 
\[\sum_{e \in \M} x_{ej} \: f_{ej}(x_{ej}) \leq f_{ij}(x_{ij}) = w_{ij}^2 + 4 \: w_{ij} \: \mathbb{E}[L_i^{(j-1)}] + x_{ij} \: w_{ij}^2 \leq w_{ij}^2 + 4 \: w_{ij} \: \mathbb{E}[L_i]\]
where the last inequality holds since $\mathbb{E}[L_i^{(j-1)}] + x_{ij} w_{ij} = \mathbb{E}[L_i^{(j)}] \leq \mathbb{E}[L_i]$.
\end{proof}

We also need a lemma about the right constants for the dual fitting.
\begin{lemma}
\label{lemma_balance_constants}
Let $\alpha, \beta \geq 0$ be defined as $\alpha = 2 \sqrt{2/5} \approx 1.265$ and $\beta = \sqrt{2/5}$. The following properties hold:
\begin{itemize}
\item $1 - \alpha^2/2 = \alpha \beta/4$
\item $\alpha \beta/4 = 1/5$
\end{itemize}
\end{lemma}
\begin{proof}
The proof is immediate.
\end{proof}We are now ready to analyze \Call{Balance}{} using our dual fitting approach. Recall that the relaxation \emph{(SDP-LB)} is this model is given by:
\vspace{0.1cm}
\begin{align*}
    \max \sum_{j \in \J} y_j  - \frac{1}{2}  &\Vert \vb \Vert ^ 2 \\
    y_j &\leq w_{ij}^2 - \frac{1}{2}\Vert v_{ij} \Vert ^ 2 + \langle \vb, v_{ij} \rangle \hspace{1.6cm} \forall j \in \J, \forall i \in \mathcal{S}_j\\
    \langle v_{ij}, v_{i'k} \rangle &\leq 2 \; w_{ij} \: w_{i'k} \: \mathds{1}_{\{i = i'\}} \hspace{2.6cm} \forall (i,j) \neq (i',k) \text{ with } j,k \in \J.
\end{align*}
\begin{theorem}
\label{thm_carag_analysis}
For any instance of the above online load balancing problem, and any fractional solution $(x_{ij})_{j \in \J, i \in \mathcal{S}_j}$ constructed by \Call{Balance}{}, there exists a feasible (SDP-LB) solution with objective value at least \[\frac{1}{5} \sum_{e \in \M} \mathbb{E} \left[L_e^2\right].\] By weak duality, this implies that the competitive ratio of \Call{Balance}{} is at most $5$.
\end{theorem}
\begin{proof}
The vectors of the SDP live in the space $\mathbb{R}^\M$. Let $\alpha, \beta \geq 0$ be defined as in Lemma \ref{lemma_balance_constants}. We now state the dual fitting:
\begin{itemize}
\item $\vb(e) := \beta \; \mathbb{E}\left[L_e\right]$
\item $v_{ij}(e) := \alpha \: w_{ej} \: \mathds{1}_{\{e = i\}}$ \qquad \qquad \qquad  $\hspace{0.8cm} \forall j \in \J, i \in \mathcal{S}_j$
\item $y_j := \frac{1}{5} \;  \sum_{e \in \M} x_{ej} \: f_{ej}(x_{ej}) \hspace{2.2cm} \forall j \in \J.$
\end{itemize}
Let us now compute the different inner products and norms that we need.
\begin{align*}
\Vert \vb \Vert ^ 2 &= \beta^2 \; \sum_{e \in \M} \: \mathbb{E}[L_e]^2 \quad \quad \hspace{0.1cm} \Vert v_{ij} \Vert ^ 2 = \alpha^2 \: w_{ij}^2 \\
\langle \vb, v_{ij} \rangle &= \alpha \beta \; w_{ij} \: \mathbb{E} [L_i] \quad \quad \langle v_{ij}, v_{i'k} \rangle = \alpha^2 \: w_{ij} \: w_{i'k} \: \mathds{1}_{\{i = i'\}}.
\end{align*}
The second set of constraints of the SDP is satisfied due to the last computation above and the fact that $\alpha^2 = 8/5 \leq 2$. The first set of constraints under the above fitting reads:
\begin{align*}
y_j \leq w_{ij}^2 - \frac{1}{2}\Vert v_{ij} \Vert ^ 2 & + \langle \vb, v_{ij} \rangle  \\
\iff \frac{1}{5} \;  \sum_{e \in \M} & x_{ej} \: f_{ej}(x_{ej}) \leq \left(1 - \frac{\alpha^2}{2} \right) w_{ij}^2 + \frac{\alpha \beta}{4} \; 4 \: w_{ij} \:  \mathbb{E}[L_i].
\end{align*}
Observe that $1 - \alpha^2/2 = \alpha \beta /4 = 1/5$ by Lemma \ref{lemma_balance_constants}, meaning that this set of constraints is now satisfied by Lemma \ref{lemma_balance_ineq_LB}. To argue about the objective function, due to \eqref{eq_obj_ineq_caralgo}, we get:
\begin{align*}
\sum_{j \in \J} y_j = \frac{1}{5} \sum_{e \in \M} \sum_{j \in \J} \left( \delta_e^{(j)}(x) + \Delta_e^{(j)}(x) \right) = \frac{1}{5} \sum_{e \in \M} \left( \mathbb{E}\left[L_e\right]^2 + \mathbb{E}\Big[L_e^2\Big] \right) 
\end{align*}
by the definition of the moment increments in \eqref{eq_increments_moments} and the fact that the sum is telescoping. Hence, the objective function is equal to
\[\sum_{j \in \J} y_j - \frac{1}{2}\Vert \vb \Vert ^ 2 = \sum_{j \in \J} y_j - \frac{\beta^2}{2} \; \sum_{e \in \M} \: \mathbb{E}[L_e]^2 = \frac{1}{5} \sum_{e \in \M} \mathbb{E} \Big[L_e^2\Big] \]
where the second equality follows from $\beta^2/2 = 1/5$ since the definition of $\beta$ states $\beta = \sqrt{2/5}$.
\end{proof}

\subsection{An optimal fractional algorithm}
In this section, we show how to get a $4$-competitive fractional algorithm. This algorithm will in fact turn out to be \emph{optimal}, with a matching lower bound presented in Section \ref{sec:lower_bound_frac}. The algorithm is still of water-filling type but with the following potential functions $f_{ij}:[0,1] \to \mathbb{R}$ for every $i \in \mathcal{S}_j$:
\begin{equation}
\label{fractional_potential}
f_{ij}(t) := w_{ij} \left(2 \: L_i^{(j-1)}(x) + t \:w_{ij} \right).
\end{equation}

{\renewcommand{\baselinestretch}{1.3}\selectfont{
\begin{algorithm}[H]
    \caption{\Call{FracBalance}{}}
    \label{opt_fractional_algo}
    \begin{algorithmic}
    \When {$j \in \J$ arrives}
    \State Compute $(x_{ij})_{i \in \mathcal{S}_j}$ such that $\sum_{i \in M} x_{ij} = 1$ and $\sum_{e \in \M} x_{ej} \: f_{ej}(x_{ej}) \leq f_{ij}(x_{ij}) \qquad \forall i \in \mathcal{S}_j$
    \EndWhen
    \State \Return $x$
 \end{algorithmic}
 \end{algorithm}
 }}
We call this algorithm \Call{FracBalance}{}. At every time step $j \in \J$, let us denote the increments
\[\delta_{e}^{(j)}(x) = L_e^{(j)}(x)^2 - L_e^{(j-1)}(x)^2 \quad \forall e \in \M.\]
This is in fact exactly the same formula as for the increment of the square of the first moment in the randomized integral case. However, the key difference here is that the total cost is now also determined by summing up these increments:
\[\sum_{e \in \M}L_e(x)^2 = \sum_{e \in \M} \sum_{j \in \J} \delta_{e}^{(j)}(x).\]
Moreover, note that we can relate this increment to \eqref{fractional_potential} as follows:
\begin{equation}
\label{eq_obj_frac_algo}
\delta_e^{(j)}(x) = \left(L_e^{(j-1)}(x) + w_{ej}x_{ej}\right)^2 - L_e^{(j-1)}(x)^2 = x_{ej}\:w_{ej} \: \left(2\: L_e^{(j-1)}(x) + x_{ej}w_{ej}\right) = x_{ej} \: f_{ej}(x_{ej}).
\end{equation}
\begin{lemma}
\label{lemma_balance_ineq}
For any solution $(x_{ij})_{j \in \J, i \in \mathcal{S}_j}$ constructed by \Call{FracBalance}{}, the following inequalities are satisfied for every $j \in \J$:
\[\sum_{e \in \M} x_{ej} \: f_{ej}(x_{ej}) \leq 2 \: w_{ij} \: L_i(x) \qquad \forall i \in \mathcal{S}_j.\]
\end{lemma}
\begin{proof}
By the equilibrium condition of the algorithm, for every $i \in \mathcal{S}_j$, we get:
\[\sum_{e \in \M} x_{ej} \: f_{ej}(x_{ej}) \leq f_{ij} (x_{ij}) = w_{ij} \left(2 \: L_i^{(j-1)}(x) + x_{ij} \:w_{ij} \right) \leq 2 \: w_{ij} \: L_i^{(j)}(x) \leq 2 \: w_{ij} \: L_i(x). \qedhere\]
\end{proof}

\begin{theorem}
For any instance and any solution $(x_{ij})_{j \in \J, i \in \mathcal{S}_j}$ obtained by \Call{FracBalance}{}, there exists a feasible (SDP-LB) solution with objective value at least \[\frac{1}{4} \sum_{e \in \M} L_e(x)^2.\] By weak duality, this implies that the competitive ratio of \Call{FracBalance}{} is at most $4$.
\end{theorem}
\begin{proof}
The vectors of the SDP live in the space $\mathbb{R}^\M$. Let $\alpha = \sqrt{2}$ and $\beta = 1/\sqrt{2}$. We now state the dual fitting:
\begin{itemize}
\item $\vb(e) := \beta \; L_e(x)$
\item $v_{ij}(e) := \alpha \: w_{ej} \: \mathds{1}_{\{e = i\}}$ \qquad \qquad \qquad  $\hspace{2cm} \forall j \in \J, i \in \mathcal{S}_j$
\item $y_j := \frac{1}{2} \;  \sum_{e \in \M} x_{ej} \: f_{ej}(x_{ej}) \hspace{3.4cm} \forall j \in \J$.
\end{itemize}
The desired inner products and norms that we need are the following.
\begin{align*}
\Vert \vb \Vert ^ 2 &= \frac{1}{2} \; \sum_{e \in \M} \: L_e(x)^2 \quad \quad \Vert v_{ij} \Vert ^ 2 = 2 \: w_{ij}^2 \\
\langle \vb, v_{ij} \rangle &=  w_{ij} \: L_i(x) \hspace{1cm} \langle v_{ij}, v_{i'k} \rangle = 2 \: w_{ij} \: w_{i'k} \: \mathds{1}_{\{i = i'\}}.
\end{align*}
The second set of constraints of the SDP is satisfied due to the last computation above. The first set of constraints under the above fitting reads:
\begin{align*}
y_j \leq w_{ij}^2 - \frac{1}{2}\Vert v_{ij} \Vert ^ 2 + \langle \vb, v_{ij} \rangle 
\iff \frac{1}{2} \;  \sum_{e \in \M} x_{ej} \: f_{ej}(x_{ej}) \leq  w_{ij} \: L_i(x).
\end{align*}
These constraints are clearly satisfied by Lemma \ref{lemma_balance_ineq}. By \eqref{eq_obj_frac_algo}, the objective function becomes:
\[\sum_{j \in \J} y_j - \frac{1}{2}\Vert \vb \Vert ^ 2 = \frac{1}{2}\sum_{e \in \M}L_e(x)^2 - \frac{1}{4} \sum_{e \in \M} L_e(x)^2 = \frac{1}{4}\sum_{e \in \M} L_e(x)^2. \qedhere\]
\end{proof}

\section{A new randomized algorithm}\label{sec_new_algo}
\label{sec_new_rand_algo}
A natural idea to improve the competitive ratio of $5$ achieved in the proof of Theorem \ref{thm_carag_analysis} would be to set $y_j = \gamma \:  \sum_{e \in \M} x_{ej} f_{ej}(x_{ej})$ for some slightly increased constant $\gamma > 1/5$. However, this could violate the first set of inequalities of the dual SDP, which, after using the equilibrium condition \eqref{equilib_cond} of the algorithm, reduces to showing:
\begin{equation}
\label{eq_high_level_dual_feas}
\gamma \: f_{ij} (x_{ij}) \leq w_{ij}^2 - \frac{1}{2}\Vert v_{ij} \Vert ^ 2 + \langle \vb, v_{ij} \rangle \hspace{1.6cm} \forall j \in \J, \forall i \in \mathcal{S}_j.
\end{equation}
After a more careful dual fitting by optimizing the dual variables (see \cref{sec_dual_feasibility}), it turns out that, for a given job-machine pair $(i,j)$, a key parameter becomes
\begin{equation}
\label{def_qij}
q_{ij} := \frac{\vb^{(j-1)}(i)}{w_{ij}}
\end{equation}
where $\vb^{(j-1)}$ denotes the state of the dual vector $\vb$ at the arrival of job $j$. In particular, inequality \eqref{eq_high_level_dual_feas} becomes tight for the \Call{Balance}{} algorithm when $q_{ij} = 2\sqrt{2/5}\approx 1.265$ and $x_{ij}=0$. This prevents us from significantly increasing $\gamma$ above $1/5$ for pairs $(i,j)$ with $q_{ij}$ close to $2\sqrt{2/5}$ and with $x_{ij}$ being close to zero.
In \cref{sec_lower_bound_ind_rounding}, we in fact construct an adversarial instance for which almost all jobs are assigned to machines with $q_{ij}\approx 2\sqrt{2/5}$ and $x_{ij}\approx 0$. We show that for this instance, no algorithm based on independent rounding can achieve a competitive ratio better than $5$. Consequently, the $5$-competitive \Call{Balance}{} algorithm by \cite{caragiannis2008better} presented and analyzed in \cref{sec_ind_rounding_algo} is \emph{optimal} among algorithms making independent random choices.

\paragraph*{Hard and easy jobs.}
To overcome this barrier, we will use a rounding method with stronger negative correlation guarantees for these problematic jobs. We will choose some interval $[a,b]\ni 2\sqrt{2/5}$ and some small threshold $\theta >0$, and we will call jobs with $q_{ij}\in [a,b]$ and $x_{ij}\in [0,\theta]$ \emph{hard jobs} for machine $i$. 
Our rounding procedure will then introduce negative correlation between the random variables $X_{ij}$ and $L_{i}^{(j-1)}$ for such jobs. All other jobs will be called \emph{easy jobs} for machine $i$. For these jobs, independent rounding suffices, as we could slightly increase $\gamma$ without violating inequality \eqref{eq_high_level_dual_feas}.

\paragraph*{Correlated randomized rounding.} As in \cref{sec_ind_rounding_algo}, our algorithm will deterministically construct a fractional solution using a water-filling approach.
To round this fractional solution, we will use a correlated rounding procedure, based on an algorithm provided by \cite{im2020weighted}. We will show that this algorithm, obtained in an offline setting, can be implemented online. We now state the guarantee obtained by this procedure. Suppose that for each machine $i \in \M$, there is a partition of the jobs into groups $\mathcal{G}_i = \{G_{i1}, \dots, G_{ik}\}$ and a fractional solution $x \in [0,1]^{\M \times \J}$ satisfying the following: each group $G \in \mathcal{G}_i$ has a total fractional value of at most one, meaning that 
$\sum_{j \in G} x_{ij} \leq 1$.
\begin{theorem}[\cite{im2020weighted}]
\label{thm_im_shadloo}
Given such a collection of groups and such a fractional solution $x$, there exists a randomized rounding algorithm outputting a feasible integral solution $X \in \{0,1\}^{\M \times \J}$ which, for every machine $i \in \M$, satisfies:
\begin{itemize}
\item $\mathbb{E}[X_{ij}] = x_{ij} \hspace{2cm} \forall j \in \J,$
\item $\mathbb{E}[X_{ij} X_{ik}] \leq x_{ij} x_{ik} \hspace{1cm} \forall j \neq k \in \J$.
\end{itemize}
Moreover, for any two jobs $j \neq k$ which belong to a same group $G \in \mathcal{G}_i$, one has stronger negative correlation:
\begin{equation}
\label{eq_stronger_neg_corr}
\mathbb{E}[X_{ij} X_{ik}] \leq \varphi(x_{ij}, x_{ik}) \: x_{ij} \: x_{ik} \quad \text{where} \quad \varphi(x,y) := \frac{e^x + e^y}{e + 1} \in [0,1].
\end{equation}
\end{theorem}

\noindent These stronger negative correlation properties are amenable for our load balancing problem, since the objective function can be written as:
\begin{equation}
\label{eq_primal_obj}
\sum_{i \in \M}\mathbb{E}[L_i(X)^2] = \sum_{i \in \M} \sum_{j,k \in \J} w_{ij} \: w_{ik} \: \mathbb{E}[X_{ij}X_{ik}].
\end{equation}
We will construct such a grouping $\mathcal{G}_i$ online for each machine $i$, where we will group hard jobs together, while each easy job will be assigned to a new group with itself as a singleton element. This will guarantee that the weights $w_{ij}$ and $w_{ik}$ of jobs $j,k$ in the same group $G\in \mathcal{G}_i$ are similar, which is needed to get an improved objective using the stronger negative correlation guarantee, due to the objective function \eqref{eq_primal_obj}. The full algorithm is described in \cref{full_algo}.

\begin{algorithm}
\caption{The full algorithm}
\label{full_algo}
\vspace{0.2cm}
\textbf{when} $j \in \J$ arrives \textbf{do}
\begin{enumerate}
\item Compute a fractional solution $(x_{ij})_{i \in \M}$ using \cref{algo_fractional_assignment} described in \cref{sec_fractional_assign}.
\item Update the grouping $\mathcal{G}_i$ for every $i \in \M$ using \cref{algo_grouping} described in \cref{sec_grouping}
\item Randomly round $(x_{ij})_{i \in \M}$ into $(X_{ij})_{i \in \M}$ by adapting the correlated rounding procedure of \cite{im2020weighted} -- explained in \cref{sec_rand_rounding} -- online using \cref{algo_online_rounding}
\item Update the dual SDP solution using \cref{algo_dual_updates} described in \cref{sec_dual_updates}
\end{enumerate}
\end{algorithm}
\FloatBarrier
\subsection{Generating the fractional assignment}
\label{sec_fractional_assign}
We first define the following constant which will determine the competitive ratio of the algorithm: 
\[\gamma:= \valGamma \approx \valGammaApprox.\]
Let us also define important constants which are needed to design the algorithm. We give more intuition on what they correspond to in \cref{sec_constants}:
\begin{equation}
\label{eq_constants_grouping_algo}
a:= \valA, \quad b:= \valB, \quad \theta:=\valTheta, \quad \delta:= \valDelta, \quad \beta:=\sqrt{2/5}.
\end{equation}
Let us also define $\phi_{ij} : [0,1] \to \mathbb{R}$ as
\begin{align}
\label{eq_phihat_ij}
\phi_{ij}(t):= \begin{cases} \beta + \delta &\text{if $q_{ij} \notin [a,b] \text{ or } t> \theta,$} \\
\beta & \text{else.}
\end{cases}
\end{align}

As a reminder, we have defined $q_{ij} := \vb^{(j-1)}(i)/w_{ij}$, which will be a key parameter to define hard jobs. When a job $j \in \J$ arrives, define the following increasing potential functions $f_{ij}: [0,1] \to \mathbb{R}_+$ for every $i \in \mathcal{S}_j$:
\begin{equation}
\label{eq_potentials}
f_{ij}(t) := \gamma \: \Big(w_{ij}^2 + 2 \: w_{ij} \: \mathbb{E}[L_i^{(j-1)}] \Big) + \frac{1}{2t} \left[\left( \vb^{(j-1)}(i) + w_{ij} \: t \: \phi_{ij}(t)\right)^2 - \vb^{(j-1)}(i)^2 \right].
\end{equation}
The dual vector $\vb \in \mathbb{R}^{\M}$ is initialized to be the all zeros vector and the way it is updated will be described in \cref{sec_dual_updates}, which will then also give more intuition for the second term of $f_{ij}$.

The fractional solution $(x_{ij})_{i \in \mathcal{S}_j}$ satisfying $\sum_{i \in \M} x_{ij} = 1$ is then computed using a \emph{water-filling} algorithm which ensures that:
\[x_{ej} > 0 \implies  f_{ej}(x_{ej}) \leq f_{ij}(x_{i j}) \qquad \forall i \in \mathcal{S}_j.\]
This algorithm continuously increases the lowest potentials until sending a total fractional value of one to achieve the above. Similarly to the previous algorithms, this implies the following equilibrium inequality:
\begin{equation}
\label{eq_equilibrium_condition_sec}
\sum_{e \in \M} x_{ej} \:  f_{ej}(x_{ej}) \leq f_{ij}(x_{ij}) \qquad \forall i \in \mathcal{S}_j.
\end{equation}
The subroutine is summarized in \cref{algo_fractional_assignment}.

{\renewcommand{\baselinestretch}{1.3}\selectfont{
\begin{algorithm}
    \caption{Water-filling to get fractional assignment}
    \label{algo_fractional_assignment}
    \begin{algorithmic}
    \When{$j \in \J$ arrives}
    \State Compute $(x_{ij})_{i \in \mathcal{S}_j}$ such that $\sum_{i \in M} x_{ij} = 1$ and $\sum_{e \in \M} x_{ej} \:  f_{ej}(x_{ej}) \leq f_{ij}(x_{ij}) \qquad \forall i \in \mathcal{S}_j$
    \EndWhen
    \State \Return $(x_{ij})_{i \in \mathcal{S}_j}$  and  $(f_{ij}(x_{ij}))_{i \in \mathcal{S}_j}$
 \end{algorithmic}
 \end{algorithm}
}}
\FloatBarrier

\subsection{The grouping procedure}
\label{sec_grouping}
Let us first state the definition of \emph{easy} and \emph{hard jobs} for each machine $i \in \M$.
\begin{definition}
\label{def_easyhardjobs}
When a job $j \in \J$ arrives and a fractional assignment $(x_{ij})_{i \in \mathcal{S}_j}$ is generated, $j$ is defined to be \emph{hard} for $i \in M$ if
\begin{equation}
\label{eq_grouping_properties}
x_{ij} < \theta \qquad \text{and} \qquad q_{ij}:= \frac{\vb^{(j-1)}(i)}{w_{ij}} \in [a,b].
\end{equation}
If one of the two conditions above is violated, then $j$ becomes \emph{easy} for $i \in M$.
\end{definition}
\begin{remark}
When the machine $i \in M$ is clear from the context, we will sometimes simply say that a job is \emph{easy}/\emph{hard}.
\end{remark}

For each machine $i \in \M$, our algorithm will maintain a partition $\mathcal{G}_i$ of the arrived jobs, such that $\sum_{j\in G}x_{ij}\leq 1$ holds for all $G\in \mathcal{G}_i$. The correlated rounding procedure that we use will induce negative correlation between the random variables $X_{ij}$ and $X_{ij'}$ whenever $j$ and $j'$ are in the same set of the partition $\mathcal{G}_i$. If $j$ and $j'$ are in different sets of the partition, the variables $X_{ij}$ and $X_{ij'}$ will be independent.

The grouping $\mathcal{G}_i$ consists of two types of sets that we call \emph{groups}. For each easy job $j \in \J$, we create a singleton group that solely contains $j$, meaning that $\{j\} \in \mathcal{G}_i$. The second type of sets are called the \emph{hard} groups, and these partition the hard jobs for $i$. Let us denote these sets as $\{G_{i,1}, G_{i,2}, \dots\}$. All arriving hard jobs are initially assigned to $G_{i,1}$. When a hard group $G_{i,k}$ for $k \geq 1$ becomes \emph{full}, meaning that $\sum_{j\in G_{i,k}}x_{ij} >  1-\theta$ holds, a new hard group $G_{i,k+1}$ is created and the following arriving hard jobs will now be assigned to this group. Observe that this maintains the invariant that $\sum_{j\in G}x_{ij}\leq 1$ for all $G\in \mathcal{G}_i$. The grouping procedure is summarized in \cref{algo_grouping}.

\begin{figure}[t]
\center
\includegraphics[width = 0.9\textwidth]{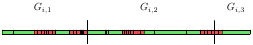}
\caption{An example of the grouping procedure for a fixed machine $i \in \M$. The length of a job $j \in \J$ reflects its fractional value $x_{ij} \in [0,1]$. Each easy job (in green) is contained in its own group, while the hard jobs (in red) are partitioned in the groups $G_{i,1},G_{i,2}$ and $G_{i,3}$. The large vertical lines indicate once a hard group becomes full. The first two hard groups $G_{i,1}$ and $G_{i,2}$ are full, while $G_{i,3}$ is not yet full: the next arriving hard job would thus be assigned to it.}
\label{fig_grouping}
\end{figure}

{\renewcommand{\baselinestretch}{1.3}\selectfont{
\begin{algorithm}
    \caption{Grouping procedure for a single machine $i \in \M$}
    \label{algo_grouping}
    \begin{algorithmic}
    \State \textbf{Initialize:} $G_{i,1}=\emptyset,\; \mathcal{G}_i=\{G_{i,1}\}, k=1$
   \When{$j \in \J$ arrives}
     \State Compute the fractional assignment $(x_{ej})_{e \in \mathcal{S}_j}$ via Algorithm \ref{algo_fractional_assignment}
     \State  Check whether $j$ is \emph{easy} or \emph{hard} for $i$ via \eqref{eq_grouping_properties}
     \If{$j$ is \emph{easy}} 
     \State Update $\mathcal{G}_i = \mathcal{G}_i \cup \{\{j\}\}$
     \Else
     \State  Update $G_{i,k} = G_{i,k} \cup \{j\}$
      \If{$\sum_{j' \in G_{i,k}} x_{ij'} > 1 - \theta$}  
      \State $G_{i,k+1} = \emptyset$
      \State $\mathcal{G}_i = \mathcal{G}_i\cup \{G_{i,k+1}\}$
      \State $k =k+1$
      \EndIf
     \EndIf
     \EndWhen
 \end{algorithmic}
 \end{algorithm}
}}

\begin{definition}
A hard group $G_{i,k} \in \mathcal{G}_i$ is defined to be \emph{full} if it satisfies $\sum_{j \in G_{i,k}}x_{ij} > 1 - \theta$.
\end{definition}
\begin{remark}Observe that, at any point in time, all hard groups except for the last one are full.
\end{remark}
The last job of every full group will turn out to be important, and we thus introduce the following notation $\mathcal{L}_i$ for those jobs. We also introduce a notation to have access to the group an arbitrary job $j \in \J$ belongs to on machine $i \in M$.
\begin{definition}
For every $i \in \M$, the set consisting of the last job of every \emph{full} hard group $G_{i,k} \in \mathcal{G}_i$  for all $k \geq 1$ is denoted by $\mathcal{L}_i \subseteq \J$.  Moreover, for every hard job $j \in \J$, we denote by $G_i(j) \in \mathcal{G}_i$ the group it belongs to on machine $i \in \M$.  
\end{definition}

\subsection{The randomized correlated rounding}
\label{sec_rand_rounding}
We present here the dependent rounding scheme of \cite{im2020weighted}. We first describe how their algorithm works in an offline setting and then briefly state how to adapt it online, leaving the details of the online implementation to Appendix \ref{sec_online_dep_rounding}. Let us define the distribution $\widetilde{\operatorname{Poi}}(p)$ for a parameter $p > 0$, as done in \cite{im2020weighted}. A random variable $X$ sampled from this distribution takes values in the integers $\{0,1,2 \dots\}$ and has the following probability mass function:
\[\mathbb{P}[X = k] = \begin{cases} e^{-p} \: p^{k-1} / k! \hspace{2cm} \text{if } k > 0,\\ 1 - (1 - e^{-p})/p \hspace{1.4cm} \text{if } k = 0. \end{cases}\]
This random variable differs from a standard Poisson variable by a factor $p$ for every $k > 0$. In particular, note that for $X\sim \widetilde{\operatorname{Poi}}(p)$ and $Y\sim \operatorname{Bernoulli}(p)$, we have $X\cdot Y \sim \operatorname{Poi}(p)$. We now describe the offline rounding algorithm, assuming that there is a grouping $\mathcal{G}_i$ for each machine $i \in M$ and a fractional solution $x$ satisfying $\sum_{j \in G} x_{ij} \leq 1$ for every $G \in \mathcal{G}_i$.

\begin{algorithm}
\caption{Offline dependent rounding algorithm by \cite{im2020weighted}}
\label{algo_offline_rounding}
\vspace{0.2cm}
\textbf{Input:} The fractional assignment $x$ and the grouping $\mathcal{G}_i$ for each $i \in M$.
\begin{enumerate}
\item The jobs are assigned in rounds $\ell \in \{1, 2, \dots\}$. We denote the set of jobs unassigned at the beginning of round $\ell \in \mathbb{N}$ by $\J_{\ell} \subseteq \J$, hence $\ell = 1$ and $\J_1 = \J$ at initialization.
\item For each $i \in \M$, each group $G \in \mathcal{G}_i$ independently \emph{recommends} at most one job $j \in G$ with probability $x_{ij}$ (and thus recommends no job with probability $1 - \sum_{k \in G}x_{ik}$). Denote by $B_{ij} \in \{0,1\}$ the indicator random variable of whether $j$ has been recommended by its group on machine $i$.
\item For each $i \in \M$ and $j \in \J$, sample $\tilde{N}_{ij} \sim \widetilde{\operatorname{Poi}}(x_{ij})$ if $x_{ij} > 0$ and set $\tilde{N}_{ij} = 0$ if $x_{ij} = 0$. Define $N_{ij} := B_{ij} \: \tilde{N}_{ij}$, this random variable counts the number of \emph{tickets} generated for $j$ by $i$. Note that $N_{ij} > 0$ only if $j$ is recommended by its group on machine $i$.
\item For each $j \in \J$, if $S := \sum_{i \in M} N_{ij} > 0$, assign $j$ to machine $i$ with probability $N_{ij}/S$, meaning that one of the generated tickets for $j$ is picked uniformly at random.
\item For each assigned job $j$ in this iteration $\ell > 0$, update $J_{\ell} = J_{\ell} \setminus \{j\}$. Moreover, for each $i \in \M$, update $G = G \setminus J_{\ell}$ for every $G \in \mathcal{G}_i$, increase $\ell = \ell + 1$, and go back to step 2 if there are unassigned jobs remaining.
\end{enumerate}
\end{algorithm}

In order to implement \cref{algo_offline_rounding} online, each group $G \in \mathcal{G}_i$ needs to recommend each job $j \in \mathcal{G}_i$ with probability $x_{ij}$ online for every round $\ell \in \{1,2, \dots\}$. This can be achieved by sampling uniform random variables $R_{G, \ell} \in [0,1]$ for every $\ell \geq 1$ as soon as a group $G \in \mathcal{G}_i$ is created online. When a job $j \in \J$ arrives, this job is recommended by its group $G$ on machine $i$ at round $\ell$ if $R_{G, \ell} \in [0,x_{ij}]$. This random variable is then updated as $R_{G, \ell} = R_{G, \ell} - x_{ij}$ and the algorithm iterates this process -- by increasing the round $\ell$ -- until job $j$ is assigned to some machine $i$. Details of the online algorithm and a proof of its equivalence with \cref{algo_offline_rounding} are shown in Appendix \ref{sec_online_dep_rounding}. This algorithm satisfies the properties in \Cref{thm_im_shadloo}.

\section{The dual updates and analysis of the algorithm}
\label{sec_analysis}
In this section, we describe the dual updates of the algorithm, prove the competitive ratio guarantee, and show dual feasibility. Note that we only need to describe how the dual vector $\vb \in \mathbb{R}^{\M}$ evolves over time to have a complete description of the algorithm, since it is needed to generate the fractional assignment, see \eqref{eq_potentials}. The other dual variables $y_j \in \mathbb{R}$ for every $j \in \J$ and $v_{ij} \in \mathbb{R}^{\M}$ for every $j \in \J, i \in \mathcal{S}_j$ will only be needed for the analysis. We now describe a high-level view of the ideas used for analyzing our main algorithm.

\paragraph*{Dual fitting.} Observe that the dual fitting in \cref{sec_ind_rounding_algo} for the \Call{Balance}{} algorithm can be interpreted as satisfying $\vb^{(j)}(i) = \beta \: \mathbb{E}[L^{(j)}_i]$ for all $(i,j)$, meaning that this dual variable is updated as $\vb^{(j)}(i) = \vb^{(j-1)}(i)+\beta \: w_{ij}x_{ij}$ at every time step. In this case, we keep this same update if  job $j$ is \emph{hard} for machine $i$.  However, if $j$ is \emph{easy} for $i$, we do a bigger increase and set $\vb^{(j)}(i)=\vb^{(j-1)}(i)+ (\beta+\delta) w_{ij}x_{ij}$ for some constant $\delta>0$. Moreover, at the end of each \emph{full} hard group $G \in \mathcal{G}_i$, we add an extra increase $B(G) > 0$ to $\vb^{(j)}(i)$, which we call a \emph{bonus}. These two modifications in the dual fitting help maintain a key invariant needed to argue dual feasibility, proven in Section \ref{sec_linking_dual_vector}, which states that $\vb^{(j)}(i) \geq (\beta + \maclambda) \: \mathbb{E}[L^{(j)}_i]$ for some $0<\maclambda<\delta$ for every $(i,j)$. This $\maclambda$ parameter will create some room to increase $\gamma$ above $1/5$ in \eqref{eq_high_level_dual_feas} for hard jobs, for which this constraint was (close to) tight.

\paragraph*{Competitive ratio.} This bonus added at the end of each hard group to $\vb^{(j)}(i)$ hurts the objective function of the SDP, since the latter is $\sum_{j \in \J} y_j - \frac{1}{2}\Vert \vb \Vert ^ 2$. However, this loss is counterbalanced by the smaller cost the algorithm incurs due to the stronger negative correlation guarantee of the new randomized rounding inside of each hard group. This bonus will thus be chosen in a group-dependent way in order to guarantee that $\sum_{j \in \J} y_j - \frac{1}{2}\Vert \vb \Vert ^ 2 \geq \gamma \: \sum_{i \in M} \mathbb{E}[L_i^2]$ at all times for $\gamma > 1/5$. The argument is formally presented in \cref{sec_SDP_obj_guarantee}.

\paragraph*{Outline.} In \cref{sec_constants}, we define some constants which have been numerically optimized and inequalities that they satisfy. In \cref{sec_dual_updates}, we define how the SDP dual solution evolves after every job arrival. In \cref{sec_SDP_obj_guarantee}, we prove the competitive ratio guarantee. In \cref{sec_linking_dual_vector}, we show the key invariant $\vb^{(j)}(i) \geq (\beta + \maclambda) \: \mathbb{E}[L^{(j)}_i]$ for every $(i,j)$ needed to argue dual feasibility. Finally, we prove dual feasibility in \cref{sec_dual_feasibility}.

\subsection{Constants and inequalities}
\label{sec_constants}
In order to analyze this algorithm, we need several constants which have been numerically optimized. In this section, we state the definition of these constants and inequalities that they satisfy which will be needed in different proofs in \cref{sec_SDP_obj_guarantee} and \cref{sec_linking_dual_vector}. First, the parameter determining the competitive ratio is denoted as
\[\gamma:= \valGamma \approx \valGammaApprox.\]
The inverse ratio of this expression yields a competitive ratio of $\valCompRatio$. Recall first the three constants which were needed to define hard jobs:
\begin{equation}
\label{eq_constants_hard_jobs}
a:= \valA, \quad b:= \valB, \quad \theta:=\valTheta.
\end{equation}
The next three constants will be needed to define the dual updates to the SDP vector $\vb$ which will be presented in \cref{sec_dual_updates}:
\begin{align*}
\beta&:=\sqrt{2/5}, \quad \delta:= \valDelta, \quad \macpsi:=\valPsi.
\end{align*}
Observe that this $\beta$ is the same as in Lemma \ref{lemma_balance_constants} needed for the independent rounding algorithm.
Finally, we also need the next three constants to prove the invariant $\vb^{(j)}(i) \geq (\beta + \maclambda) \: \mathbb{E}[L^{(j)}_i]$ in \cref{sec_linking_dual_vector} needed to show dual feasibility:

\begin{equation}
\label{eq_constants_dualfeas}
\maclambda:=\valLamb, \quad \maceta:=\valEta, \quad \tauone:= \valTau.
\end{equation}

\begin{proposition}
Let us define the following three expressions for convenience of notation:
\[\kappa := \frac{\exp(\beta /a)}{1 - \tauone \exp(\beta /a)}; \quad K:= \frac{1}{\beta}\left(\exp\left(\frac{\beta }{a}\right)-1 \right); \quad \omega:=\frac{\maceta}{1/K+\beta +\maceta}.\]
The following inequalities are then satisfied:
\begin{align}
\macpsi + \frac{\macpsi^2}{2} &\leq \frac{\gamma}{b^2} \Big(1-\varphi(\theta,\theta)\Big)(1-2\theta)(1 - \theta) \label{eq1} \\
\kappa(\kappa -1) \maceta &\leq \macpsi \beta \label{eq2}\\
\frac{\beta+\maceta}{\beta+\delta}+\frac{\maceta}{\tauone a} &\leq 1 \label{eq3}\\
\beta + \maclambda &\leq (1-\omega)(\beta + \maceta). \label{eq4}
\end{align}
\end{proposition}
\noindent Several other satisfied inequalities will be needed to argue dual feasibility in the proof of \cref{lemma_dual_feas} presented in \cref{sec_dual_feasibility}, which we omit here. We have numerically optimized the values of all the above parameters to maximize $\gamma$ while keeping all these different constraints satisfied.

\subsection{The dual updates}
\label{sec_dual_updates}
Let us first rewrite the dual program that we are using. We want to use this dual program to generate a feasible dual solution with the largest possible objective value to get the best possible bound on the competitive ratio.

\begin{align}
    \max \sum_{j \in \J} y_j  - \frac{1}{2}  &\Vert \vb \Vert ^ 2 \label{obj} \\
    y_j &\leq w_{ij}^2 - \frac{1}{2}\Vert v_{ij} \Vert ^ 2 + \langle \vb, v_{ij} \rangle \hspace{1.4cm} \forall j \in \J, \forall i \in \mathcal{S}_j \label{first_constr}\\
    \langle v_{ij}, v_{i'k} \rangle &\leq 2 \; w_{ij} \: w_{i'k} \: \mathds{1}_{\{i = i'\}} \hspace{2.4cm} \forall (i,j) \neq (i',k) \text{ with } j,k \in \J \label{second_constr}.
\end{align}

Whenever $j$ arrives online, we thus have a variable $y_j \in \mathbb{R}$ and a collection of vectors $v_{ij} \in \mathbb{R}^\M$ for every $i \in \mathcal{S}_j$ which arrive as well. We will also see the vector $\vb \in \mathbb{R}^\M$ as monotonically increasing online in each coordinate over time. At the arrival of $j \in \J$, this vector is in a state $\vb^{(j-1)} \in \mathbb{R}^\M$. 

\paragraph*{Fitting the vectors $v_{ij}$.} A natural way to fit the vectors $v_{ij}$ is to set 
\begin{equation}
\label{eq_v_ij_fitting}
v_{ij}(e) = \alpha_{ij} \: w_{ij} \: \mathds{1}_{\{i = e\}} \quad \forall e \in \M
\end{equation}
for some $\alpha_{ij} \geq 0$ in order to always have a vector with support one. This will ensure that constraints \eqref{second_constr} are always satisfied as long as $\alpha_{ij} \leq \sqrt{2}$. Under the above fitting, constraints \eqref{first_constr} become:
\begin{equation}
\label{eq_dual_feas_sufficient_condition}
y_j \leq \left(1 - \frac{\alpha_{ij}^2}{2}\right) w_{ij}^2 + \alpha_{ij} \: w_{ij} \: \vb(i) \qquad \forall j \in \J, \forall i \in \mathcal{S}_j.
\end{equation}
Clearly, we want to pick $\alpha_{ij}$ to have the right-hand side above to be as large as possible, while still ensuring that $\alpha_{ij} \leq \sqrt{2}$. By taking the derivative of the right-hand side above with respect to $\alpha_{ij}$ and setting it to zero, we would have it equal to $\vb(i)/w_{ij}$. However, at the time job $j$ arrives, the online algorithm is only aware of $\vb^{(j-1)}$. A natural way to pick $\alpha_{ij}$ is thus to fit it as follows:

\begin{equation}
\label{eq_alphaij}
\alpha_{ij} = \min \left\{\frac{\vb^{(j-1)}(i)}{w_{ij}}, \sqrt{2}\right\} = \min \{q_{ij}, \sqrt{2}\}.
\end{equation}
Observe that this is where the key parameter $q_{ij}$ defined in \eqref{def_qij} comes up.

\paragraph*{Fitting the scalar $y_j$.} Similarly to the previous algorithms in this paper, our fractional water-filling algorithm satisfies the following equilibrium condition for every $j \in \J$, see \eqref{eq_equilibrium_condition_sec}:
\begin{equation}
\label{eq_equilib}
\sum_{e \in \M} x_{ej} \: f_{ej}(x_{ej}) \leq f_{ij}(x_{ij}) \qquad \forall i \in \mathcal{S}_j.
\end{equation}
The $y_j$ variable is then fitted to be:
\begin{equation}
\label{eq_yfitting}
y_j = \sum_{e \in \M} x_{ej} \: f_{ej}(x_{ej}).
\end{equation}
By the above fitting, \eqref{eq_equilib}, \eqref{eq_dual_feas_sufficient_condition} and $\vb^{(j)}(i) \leq \vb(i)$, showing the following lemma would thus be enough to prove dual feasibility.
\begin{lemma}
For every $j \in \J$ and every $i \in \mathcal{S}_j$, the following is satisfied:
\begin{equation}
\label{eq_SDP2}
f_{ij}(x_{ij}) \leq \left(1 - \frac{\alpha_{ij}^2}{2}\right) w_{ij}^2 + \alpha_{ij} \: w_{ij} \: \vb^{(j)} (i).
\end{equation}
This implies that the constructed dual solution is feasible.
\end{lemma}
\noindent This will be proven in Section \ref{sec_dual_feasibility}. We now describe how the vector $\vb \in \mathbb{R}^\M$ evolves over time.

\paragraph*{Fitting the vector $\vb$.} The vector $\vb \in \mathbb{R}^\M$ is initialized to be the all zeros vector. We will also keep track of a slightly modified vector $\vh \in \mathbb{R}^\M$, also initialized at zero. For every machine $i \in \M$, this vector will satisfy $\vh^{(j)}(i) = \vb^{(j)}(i)$ for every $j \notin \mathcal{L}_i$ and $\vh^{(j)}(i) < \vb^{(j)}(i)$ for every $j \in \mathcal{L}_i$. As a reminder, $\mathcal{L}_i \subseteq \J$ denotes the set of jobs which are the last in every full hard group $G \in \mathcal{G}_i$.

In the three algorithms presented in the previous sections, one may think about the dual vector being updated as $\vb^{(j)}(i) = \vb^{(j-1)}(i) + \beta \:  w_{ij} x_{ij}$ for every job $j$, since we always fitted $\vb(i) = \beta \: \mathbb{E}[L_i]$. In this case, we will update it more aggressively by using $\phi_{ij}(x_{ij}) \geq \beta$, which is defined in \eqref{eq_phihat_ij}, an expression which can be written in a simpler way as follows due to the definition of \emph{easy} and \emph{hard} jobs given in \eqref{def_easyhardjobs}, where we can also remove the dependence on $x_{ij}$ for simplicity of notation:
\begin{align}
\label{eq_phihat_ij_x}
\phi_{ij}:= \begin{cases} \beta + \delta & \text{if $j$ \text{is easy for} $i$,} \\
\beta   &\text{if $j$ \text{is hard for} $i$.}
\end{cases}
\end{align}

{\renewcommand{\baselinestretch}{1.35}\selectfont{
\begin{algorithm}
    \caption{Updating the dual}
    \label{algo_dual_updates}
    \begin{algorithmic}
    \When{$j \in \J$ arrives}
    \State Compute $(x_{ij})_{i \in \mathcal{S}_j}$ and $(f_{ij}(x_{ij}))_{i \in \mathcal{S}_j}$ via Algorithm \ref{algo_fractional_assignment}
    \State Set $y_j = \sum_{e \in \M} x_{ej} \: f_{ej}(x_{ej})$
    \For {$i \in \mathcal{S}_j$}
    	\State Set $\alpha_{ij} = \min \left\{\vb^{(j-1)}(i)/w_{ij}, \sqrt{2}\right\}$
	\State Set $v_{ij}(e) = \alpha_{ij} \: w_{ij} \: \mathds{1}_{\{i = e\}}$ for every $e \in \M$
	\State Check whether $j$ is \emph{hard/easy} for $i$ via \eqref{eq_grouping_properties} and compute $\phi_{ij}$ via \eqref{eq_phihat_ij_x}
	\State Set $\vh^{(j)}(i) = \vb^{(j-1)}(i) + w_{ij} \: x_{ij} \: \phi_{ij}$
	\State Check whether $j \in \mathcal{L}_i$ and set $\vb^{(j)}(i) = \vh^{(j)}(i) + B(G_i(j)) \: \mathds{1}_{\{j \in \mathcal{L}_i\}}$
	\EndFor
	\EndWhen
 \end{algorithmic}
 \end{algorithm}
 }}
 \noindent Let us now describe how the vector $\vh$ evolves over time:

\begin{equation}
\label{v_hat}
\vh^{(j)}(i) = \vb^{(j-1)}(i) + w_{ij} \: x_{ij} \: \phi_{ij}.
\end{equation}
The dual vector $\vb$ is obtained by adding a \emph{bonus} $B(G_i(j)) > 0$ to the above expression at the end of each full hard group $G \in \mathcal{G}_i$:
\begin{equation}
\label{eq_actualv}
\vb^{(j)}(i) = \vh^{(j)}(i) + B(G_i(j)) \: \mathds{1}_{\{j \in \mathcal{L}_i\}}.
\end{equation}
For a \emph{full} hard group $G\in \mathcal{G}_i$ starting with job $s+1 \in \J$ and ending with job $j \in \mathcal{L}_i$, this bonus is formally defined as:
\begin{equation}
\label{eq_bonus}
B(G) = \macpsi \: \frac{\vb^{(s)}(i)^2}{\vh^{(j)}(i)} \qquad \text{where} \qquad \macpsi = \valPsi.
\end{equation}
We are now able to give more intuition behind the definition of $f_{ij}$ given in \eqref{eq_potentials}. In particular, the important thing to observe is that, when evaluating these functions at $x_{ij}$, we now get:
\begin{align}
\label{eq_potentialf}
f_{ij}(x_{ij}) &= \gamma \Big(w_{ij}^2 + 2 \: w_{ij} \: \mathbb{E}[L_i^{(j-1)}] \Big) + \frac{1}{2 x_{ij}} \left( \vh^{(j)}(i)^2 - \vb^{(j-1)}(i)^2 \right).
\end{align}

\subsection{The competitive ratio guarantee}
\label{sec_SDP_obj_guarantee}
In this section, we prove the following theorem.
\begin{theorem}
\label{thm_objective_guarantee}
For any instance of the online load balancing problem, the constructed primal and dual solutions by \cref{full_algo} satisfy
\begin{equation}
\label{eq_objective}
\sum_{j \in \J} y_j - \frac{1}{2} \Vert \vb \Vert ^2 \geq \gamma \sum_{i \in \M} \mathbb{E}[L_i^2] \qquad \text{where} \qquad \gamma = \valGamma.
\end{equation}
\end{theorem}
Let us rewrite the above expression in a more convenient form by using the insights from the previous sections. Note that we have:
\begin{align*}
\frac{1}{2} \: \vb(i)^2 &= \frac{1}{2} \sum_{j \in \J} \left( \vb^{(j)}(i)^2 - \vb^{(j-1)}(i)^2 \right) \\ &= \frac{1}{2} \sum_{j \in \J} \left( \vb^{(j)}(i)^2 - \vh^{(j)}(i)^2 + \vh^{(j)}(i)^2 - \vb^{(j-1)}(i)^2 \right).
\end{align*}
As a reminder, we have in \eqref{eq_yfitting} fitted $y_j = \sum_{i \in \M}x_{ij} f_{ij}(x_{ij})$. Hence, by the above equation and \eqref{eq_potentialf}, we may write the SDP objective as:
\[\sum_{j \in \J} y_j - \frac{1}{2} \Vert \vb \Vert ^2 = \gamma \: \sum_{j \in \J}\sum_{i \in \M} x_{ij} \Big(w_{ij}^2 + 2 \: w_{ij} \: \mathbb{E}[L_i^{(j-1)}] \Big) - \frac{1}{2}\sum_{i \in \M} \sum_{j \in \J} \left(\vb^{(j)}(i)^2 - \vh^{(j)}(i)^2 \right).\]
Let us now rewrite the right-hand side in \eqref{eq_objective}. Observe that
\[\mathbb{E}[L_i^2] = \sum_{j \in \J} \mathbb{E}\left[{L_i^{(j)}}^2 - {L_i^{(j-1)}}^2\right] = \sum_{j \in \J} \left(w_{ij}^2 x_{ij} + 2 \: w_{ij} \:  \mathbb{E}[L_i^{(j-1)} X_{ij}] \right)\]
since $L_i^{(j)} = L_i^{(j-1)}+X_{ij}w_{ij}$. Thus, by rearranging terms, we can now rewrite \eqref{eq_objective} as:
\begin{equation}
\label{eq_object_mod}
2 \gamma \sum_{i \in \M} \sum_{j \in \J}w_{ij} \: \left(\mathbb{E}[L_i^{(j-1)}]x_{ij} - \mathbb{E}[L_i^{(j-1)} X_{ij}]\right) \geq \frac{1}{2}\sum_{i \in \M} \sum_{j \in \J} \left(\vb^{(j)}(i)^2 - \vh^{(j)}(i)^2 \right).
\end{equation}
Since we know that $\mathbb{E}[X_{ij}X_{ik}] \leq x_{ij} x_{ik}$ for any two jobs $j,k$, the terms on the left-hand side are all non-negative. In fact, for \emph{easy} jobs, the terms on left-hand side give a contribution of zero, due to independent rounding. Now, note that $\vb^{(j)}(i)^2 - \vh^{(j)}(i)^2$ is non-zero only for a last job in a full hard group, i.e., if $j \in \mathcal{L}_i$. Fix such a full group $G \in \mathcal{G}_i$ starting with job $s+1 \in \J$ and ending with $j \in \mathcal{L}_i$, by definitions \eqref{eq_actualv} and \eqref{eq_bonus}, we now have:
\begin{align*}
\vb^{(j)}(i)^2 - \vh^{(j)}(i)^2 &= B(G_i(j))^2 + 2\: \vh^{(j)}(i) \: B(G_i(j)) \\
&= \left(\macpsi \: \frac{\vb^{(s)}(i)^2}{\vh^{(j)}(i)}\right)^2 + 2 \: \macpsi \: \vb^{(s)}(i)^2 \leq (\macpsi^2 + 2 \macpsi) \: \vb^{(s)}(i)^2
\end{align*}
where the inequality follows since $s \leq j$ and thus $\vb^{(s)}(i) \leq \vh^{(j)}(i)$. Hence, by arguing on a machine by machine basis, and a group $G \in \mathcal{G}_i$ by group basis, it is enough to show the following lemma in order to prove \eqref{eq_object_mod} -- and thus \cref{thm_objective_guarantee} -- because of the two equations above. This is where the negative correlation guarantee of Theorem \ref{thm_im_shadloo} is exploited.

\begin{lemma}
Let $i \in \M$ be a machine, let $G \in \mathcal{G}_i$ be a full hard group starting with job $s+1 \in \J$. The following holds:
\[2 \gamma \sum_{j \in G} w_{ij} \: \left(\mathbb{E}[L_i^{(j-1)}]x_{ij} - \mathbb{E}[L_i^{(j-1)} X_{ij}]\right) \geq \left(\frac{\macpsi^2}{2} + \macpsi\right) \: \vb^{(s)}(i)^2.\]
\end{lemma}
\begin{proof}
For every $j \in G$, we have that:
\begin{align*}
\mathbb{E}[L_i^{(j-1)}]x_{ij} - \mathbb{E}[L_i^{(j-1)} X_{ij}] &= \sum_{k \in \J : k< j} w_{ik} \left(x_{ij} x_{ik} - \mathbb{E}[X_{ij} X_{ik}] \right) \\
& \geq \sum_{k \in G : k< j} w_{ik} \left(x_{ij} x_{ik} - \mathbb{E}[X_{ij} X_{ik}] \right) \\
& \geq \sum_{k\in G:k<j} w_{ik} \Big(1 - \varphi(x_{ij},x_{ik})\Big)x_{ij} x_{ik}
\end{align*}
where the first inequality follows by truncating the sum to only the hard jobs in $G$, which we are allowed to do since $\mathbb{E}[X_{ij} X_{ik}] \leq x_{ij} x_{ik}$ for every pair of jobs $j \neq k$, and the second inequality follows from the negative correlation guarantee. Now, observe that due to the grouping properties \eqref{eq_grouping_properties}, we have that for every $j,k \in G$:
\[\varphi(x_{ij},x_{ik}) \leq \varphi(\theta,\theta), \quad w_{ij} \geq \frac{\vb^{(j-1)}(i)}{b} \geq \frac{\vb^{(s)}(i)}{b}, \quad w_{ik} \geq \frac{\vb^{(k-1)}(i)}{b} \geq \frac{\vb^{(s)}(i)}{b}.\]
The left-hand side in the statement of the lemma may thus be lower bounded by:
\[\frac{2\gamma}{b^2}\Big(1 - \varphi(\theta, \theta)\Big) \: \vb^{(s)}(i)^2 \sum_{j,k \in G} x_{ij} x_{ik} \mathds{1}_{\{k < j\}}.\]
Let us define $S := \sum_{j \in G} x_{ij}$, since we know that $x_{ij} \leq \theta$ for all $j \in G$, that $S \geq 1 - \theta$ since we are working in a full group, and that $\theta \leq 1/2$, we get
\[\sum_{j,k \in G} x_{ij} x_{ik} \mathds{1}_{\{k < j\}} = \frac{S^2 - \sum_{j \in G}x_{ij}^2}{2} \geq \frac{S (S - \theta)}{2} \geq \frac{(1 - \theta)(1 - 2\theta)}{2}.\]
Therefore, the left-hand side is at least
\begin{align*}
\frac{\gamma}{b^2} \Big(1-\varphi(\theta,\theta)\Big)(1-2\theta)(1 - \theta) \; \vb^{(s)}(i)^2.
\end{align*}
The lemma now follows because of inequality \eqref{eq1}.
\end{proof}

\subsection{Linking the dual vector to the expected load}
\label{sec_linking_dual_vector}

In this section, we show how $\vb^{(j)}(i)$ relates to the expected load $\mathbb{E}[L_i^{(j)}]$ -- which we will denote for clarity by $\mathbb{E}[L^{(j)}(i)]$ in this subsection -- for every job $j \in \J$ and every machine $i \in M$, which will be needed to argue dual feasibility. The important result from this section will be \cref{lem:lower_bound_v0}, which states that
\[\vb^{(j)}(i) \geq \left(\beta + \maclambda \right)\mathbb{E}[L^{(j)}(i)]\]
for a small constant $\maclambda > 0$. Note the difference with the analysis of the Caragiannis independent rounding algorithm in the proof of Theorem \ref{thm_carag_analysis}, where $\vb^{(j)}(i) = \beta \: \mathbb{E}[L^{(j)}(i)]$. 

Before proving that, we first need to do so for the jobs which precede the beginning of a hard group, for which we will be able to exploit the bonus defined in \eqref{eq_bonus}, since this bonus adds an additional increase to $\vb^{(j)}(i)$ at the end $j \in \mathcal{L}_i$ of each hard full group.
\begin{lemma}
    \label{lem:lower_bound_v0_end_group}
    Let $i \in \M$ be a machine and let $G \in \mathcal{G}_i$ be an arbitrary hard group starting with job $t+1 \in \J$. The following holds:
    \begin{align*}
        \vb^{(t)}(i) \geq \left(\beta + \maceta\right)\mathbb{E}[L^{(t)}(i)] \qquad \text{where} \qquad \maceta = \valEta.
    \end{align*}
\end{lemma}

\begin{proof}
Let us fix a machine $i \in \M$. For simplicity of notation and ease of presentation, we denote $v^{(k)}:= \vb^{(k)}(i)$ and $\hat{v}^{(k)} := \vh^{(k)}(i)$ for every job $k \in \J$ in this proof. 

We show the claim by induction for each hard group $G \in \mathcal{G}_i$. Consider the first hard group and let us assume that it starts with job $t +1 \in \J$: this means that all jobs $\{1, \dots, t\}$ are in fact easy jobs for $i$. Hence, at job $t \in \J$, we have $\vb^{(t)}(i) = (\beta + \delta) \: \mathbb{E}[L^{(t)}(i)]$ by \eqref{v_hat}, which implies the claim since $\delta > \maceta$. 

Consider now a hard group $G \in \mathcal{G}_i$ starting with job $t +1 \in \J$ such that there exists at least one full hard group $H \in \mathcal{G}_i$ preceding it, starting with job $s+1 \in \J$ and ending with job $j \in \mathcal{L}_i$. Consider the set $\{s+1, \dots, j\} \subseteq \J$: this set is composed of both the hard jobs $H$, in addition to some easy jobs that we denote by $E$, see Figures \ref{fig_proof_case_one} and \ref{fig_proof_case_two}. Let us also denote by
\[L_e = \sum_{k \in E} w_{ik} \: x_{ik} \qquad \text{and} \qquad  L_h = \sum_{k \in H} w_{ik} \: x_{ik}\]
the contribution to the expected load $\mathbb{E}[L^{(j)}(i)]$ from the easy and hard jobs respectively. Similarly, denote by
\begin{equation}
\label{eq_ve_and_vh}
v_e = (\beta + \delta)L_e  \qquad \text{and} \qquad v_h = \beta L_h
\end{equation}
the contribution to $v^{(j)}$, see \eqref{v_hat}, from the jobs in $E$ and $H$ respectively. We then have
\begin{equation}
\label{eq_vandvhat}
\hat{v}^{(j)} = v^{(s)} + v_e + v_h \quad \text{and} \quad v^{(j)} = \hat{v}^{(j)} + B(H)
\end{equation}
where $B(H)= \macpsi \: {v^{(s)}}^2 / \hat{v}^{(j)}$ as a reminder, see \eqref{eq_bonus}. In order to show the lemma, it would now suffice to show
\begin{equation}
\label{eq_sufficient_condition}
v_e + v_h + B(H) \geq \left(\beta + \maceta\right) (L_e + L_h).
\end{equation}
This would be enough to prove the statement due to the induction hypothesis $v^{(s)} \geq (\beta + \maceta) \mathbb{E}[L^{(s)}(i)]$, as well as the fact that the jobs $\{j+1, \dots, t\}$ are all easy, for which the increase is at a factor $\beta + \delta$ with $\delta > \maceta$, see \eqref{v_hat}.
The proof now splits into two cases, depending on whether $v_e$ is large compared to $\hat{v}^{(j)}$ or not.

\subparagraph{Case 1.} Suppose first that $v_e \leq \tauone \: \hat{v}^{(j)}$ for $\tauone \geq 0$ defined in \eqref{eq_constants_dualfeas}. In this case, we will use a lower bound on $B(H)$ to pay for the increase of hard jobs by arguing that 
\begin{equation}
\label{eq_case1}
v_h + B(H) \geq (\beta + \maceta)L_h.
\end{equation}
This would be enough to prove the statement of the lemma by \eqref{eq_sufficient_condition},  due to $v_e = (\beta + \delta) L_e$ with $\delta > \maceta$. By definition \eqref{v_hat} and induction on $k$ one can show that for all $ k > s$:
\begin{align}
\label{eq_upper_bound_vj}
v^{(k)}\leq \left(v^{(s)}+(\beta +\delta )\sum_{l=s+1}^{k}w_{il}x_{il}\mathds{1}_{\{l \in E\}}\right)\prod_{l=s+1}^k \Big(1+\frac{\beta}{q_{il}}  x_{il}\mathds{1}_{\{l\in H\}}\Big).
\end{align}

\begin{figure}
\center
\includegraphics[width = 0.4\textwidth]{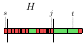}
\caption{Illustration of the first case in the proof of Lemma \ref{lem:lower_bound_v0_end_group}. The length of each job represents its contribution $w_{ij}x_{ij}$ to the expected load of machine $i$. In this case, the contribution of the easy jobs $E \subseteq \{s+1, \dots, j\}$ is small, implying that $\hat{v}^{(j)}$ is only a small constant factor bigger than $v^{(s)}$, see \eqref{eq_vj_ve_relation_case1}. This implies that the bonus added at 
$j \in H$ satisfies $B(H) = \Omega(L_h)$, see \eqref{eq_bonus_lower_bound} and \eqref{eq_Lh_upperb}, where $L_h$ denotes the total contribution from the hard jobs $H \subseteq \{s+1,\dots, j\}$.}
\label{fig_proof_case_one}
\end{figure}

The first term is clearly upper bounded by $v^{(s)} + v_e$. Moreover, by using the inequality $1+u \leq \exp(u)$, the fact that $q_{ik} \in [a,b]$ for all $k\in H$ and that  $\sum_{k \in H} x_{ik} \leq 1$, we get
\begin{align*}
\hat{v}^{(j)} \leq \left(v^{(s)}+v_e\right)\exp\left(\sum_{l=s+1}^{j}\frac{\beta}{a}x_{il}\mathds{1}_{\{l\in H\}}\right) \leq \exp\left(\frac{\beta}{a}\right) (v^{(s)} + v_e).
\end{align*}
Therefore, due to our assumption on the upper bound on $v_e$, we get
\[\hat{v}^{(j)} \leq \exp\left(\frac{\beta }{a}\right) v^{(s)} + \tauone \exp\left(\frac{\beta}{a}\right) \hat{v}^{(j)}. \]
We may rewrite the expression above as:
\begin{equation}
\label{eq_vj_ve_relation_case1}
\hat{v}^{(j)} \leq \kappa \: v^{(s)} \quad \text{where} \quad \kappa := \frac{\exp(\beta /a)}{1 - \tauone \exp(\beta /a)}.
\end{equation}
This, in turn, allows us to get a lower bound on $B(H)$:
\begin{equation}
\label{eq_bonus_lower_bound}
B(H) = \macpsi \: \frac{{v^{(s)}}^2}{\hat{v}^{(j)}} \geq \frac{\macpsi}{\kappa} \: v^{(s)}.
\end{equation}
Let us now upper bound $L_h$:
\begin{equation}
\label{eq_Lh_upperb}
L_h = \frac{1}{\beta} v_h \leq \frac{1}{\beta} \left( \hat{v}^{(j)} - v^{(s)} \right) \leq \frac{\kappa - 1}{\beta} v^{(s)}.
\end{equation}
Therefore, we now get
\[v_h + B(H) \geq \beta L_h + \frac{\macpsi}{\kappa}v^{(s)} \geq \beta L_h + \frac{\beta \macpsi}{\kappa (\kappa-1)}L_h = \left(\beta + \frac{\macpsi \beta}{\kappa (\kappa -1)}\right)L_h.\]
Due to inequality \eqref{eq2}, this shows \eqref{eq_case1}, which is what we needed.

\begin{figure}
\center
\includegraphics[width = 0.8\textwidth]{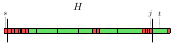}
\caption{Illustration of the second case in the proof of Lemma \ref{lem:lower_bound_v0_end_group}. In this case, the contribution of the easy jobs $E \subseteq \{s+1, \dots, j\}$ is large, meaning that it satisfies $v_e \geq \tau \hat{v}^{(j)}$. This is enough to prove the lemma without using the bonus as every easy job $k \in E$ give an increase of $(\beta + \delta) w_{ik} x_{ik}$ to the dual vector, where $\delta > \maceta$.}
\label{fig_proof_case_two}
\end{figure}

\subparagraph{Case 2.} Let us now suppose that $v_e \geq \tauone \: \hat{v}^{(j)}$. In this case, the intuition is that since the contribution of easy jobs is large, and that their increase is at a rate $\beta + \delta$ with $\delta > \maceta$, see \eqref{v_hat}, the contribution from $v_e$ itself is enough to show \eqref{eq_sufficient_condition}, without even needing the bonus $B(H)$. Formally, we have
        \begin{align*}
            L_h=\sum_{k\in H}\:x_{ik} \: w_{ik} =\sum_{k \in H}x_{ik} \frac{w_{ik}}{v^{(k-1)}} v^{(k-1)} \leq \frac{\hat{v}^{(j)}}{a} \sum_{k \in H}x_{ik} \leq  \frac{\hat{v}^{(j)}}{a}\leq \frac{v_e}{\tauone a}
        \end{align*}
        where the first inequality uses the second grouping property in \eqref{eq_grouping_properties}.
        Since $v_h = \beta L_h$, we now want to use the above lower bound on $v_e$ to argue
        \[v_e \geq (\beta + \maceta)L_e + \maceta \: L_h.\]
        However, we also have $v_e = (\beta + \delta) L_e$, meaning that
        \begin{align*}
        (\beta+\maceta) L_e +   \maceta L_h \leq \left(\frac{\beta+\maceta}{\beta+\delta}+\frac{\maceta}{\tauone a}\right)v_e \leq v_e
        \end{align*}
        where the last inequality follows from inequality \eqref{eq3}. This is enough to prove the lemma due to \eqref{eq_sufficient_condition} and $v_h = \beta L_h$.
        \end{proof}

We will now argue such a guarantee for every job $j \in \J$. The additive constant gets slightly worse compared to Lemma \ref{lem:lower_bound_v0_end_group}.
\begin{lemma}
    \label{lem:lower_bound_v0}
    For every machine $i \in \M$ and every job $j \in \J$, we have:
    \begin{align*}
        \vb^{(j)}(i) \geq \left(\beta + \maclambda \right)\mathbb{E}[L^{(j)}(i)] \quad \text{where} \quad \maclambda = \valLamb.
    \end{align*}
\end{lemma}
\begin{proof}
The proof is similar to the proof of \cref{lem:lower_bound_v0_end_group}. Let us fix a machine $i \in \M$. For simplicity of notation, we denote again $v^{(k)}:= \vb^{(k)}(i)$ and $\hat{v}^{(k)} := \vh^{(k)}(i)$ for every job $k \in \J$ in this proof. Fix a job $j \in \J \setminus \mathcal{L}_i$ and let us denote by $H \in \mathcal{G}_i$ the hard group which is not full yet at time step $j$, whose first job we denote by $s + 1 \in \J$. Consider the set $\{s + 1, \dots j\} \subseteq \J$, similarly to the proof of the previous lemma (see \eqref{eq_ve_and_vh}), we split this set into $E \cup H$, where $E$ consists of easy jobs, and we denote 
\[L_e = \sum_{k \in E} w_{ik} \: x_{ik} \qquad L_h = \sum_{k \in H} \: w_{ik} \: x_{ik} \qquad L^{(s)} = \mathbb{E}[L^{(s)}(i)].\]
As before, we also denote $v_e = (\beta + \delta) L_e$ and $v_h = \beta L_h$. Using these definitions, we now have
\begin{equation}
\label{eq_vj_and_ej}
v^{(j)} = v^{(s)} + v_e + v_h \qquad \text{and} \qquad \mathbb{E}[L^{(j)}(i)] = L^{(s)} + L_e + L_h.
\end{equation}
In contrast to \eqref{eq_vandvhat} in the previous lemma, we do not have the bonus term $B(G_i(j))$ to help increase $v^{(j)}$. However, we do know that $v^{(s)} \geq (\beta+ \maceta) \mathbb{E}[L^{(s)}(i)]$ by Lemma \ref{lem:lower_bound_v0_end_group}, and we will use this extra gain to lower bound $v^{(j)}$. This is the reason the additive constant in this lemma gets slightly worse compared to Lemma \ref{lem:lower_bound_v0_end_group}. Using the same computations as in \eqref{eq_upper_bound_vj} of the previous lemma, we have:
\[v^{(j)} \leq \exp\left(\frac{\beta}{a}\right) (v^{(s)} + v_e).\]
By $v_h = v^{(j)} - (v^{(s)} + v_e)$, and the above equation, we get
\begin{equation}
\label{eq_upper_bound_Lj}
L_h = \frac{v_h}{\beta} \leq K (v^{(s)} + v_e) \qquad \text{where} \qquad K:= \frac{1}{\beta}\left(\exp\left(\frac{\beta }{a}\right)-1 \right).
\end{equation}
We now have two different lower bounds on $v^{(s)} + v_e$:
\begin{align*}
    v^{(s)}+v_e\geq \frac{1}{K} \: L_h \qquad \text{and} \qquad v^{(s)}+v_e \geq (\beta+ \maceta) \left(L^{(s)} + L_e\right)
\end{align*}
where the second one follows from Lemma \ref{lem:lower_bound_v0_end_group} and $v_e = (\beta + \delta)L_e > (\beta + \maceta)L_e$.
By taking a convex combination of these bounds, for every $\omega\in [0,1]$, we have:

\begin{align}
\label{eq_vj_bounds}
v^{(j)} = v^{(s)} + v_e + v_h &\geq 
(1-\omega)(\beta + \maceta)\left(L^{(s)}+L_e\right) + \omega \: \frac{1}{K} L_h + \beta L_h .
\end{align}
We now ensure that 
\[(1-\omega)(\beta + \maceta)=\frac{\omega}{K}+ \beta \qquad \text{by setting } \qquad \omega=\frac{\maceta}{1/K+\beta +\maceta}.\]
Finally, since $(1-\omega)(\beta + \maceta) \geq \beta + \maclambda$ in constraint \eqref{eq4}, we get $v^{(j)} \geq (\beta + \maclambda) (L^{(s)} + L_e + L_h)$ by \eqref{eq_vj_bounds}, therefore proving the lemma. 
\end{proof}

\subsection{Dual feasibility}
\label{sec_dual_feasibility}
In this section, we show that the constructed dual SDP solution is feasible and thus complete the analysis of our algorithm.
\begin{theorem}
For any online instance, the dual solution constructed by \cref{algo_dual_updates} is a feasible solution to the semidefinite program (SDP-LB). Together with \cref{thm_objective_guarantee}, this shows that \cref{full_algo} achieves a competitive ratio of $\valCompRatio$.
\end{theorem}
\noindent In order to show dual feasibility, we have seen in \eqref{eq_SDP2} that it is enough to prove 
\begin{equation}
\label{eq_SDP2_restate}
f_{ij}(x_{ij}) \leq \left(1 - \frac{\alpha_{ij}^2}{2}\right) w_{ij}^2 + \alpha_{ij} \: w_{ij} \: \vb^{(j)} (i) \qquad \forall j \in \J, \forall i \in \mathcal{S}_j.
\end{equation}
Let us fix a job $j \in \J$ and a machine $i \in \M$, we have seen that:
\begin{itemize}
\item $f_{ij}(x_{ij}) = \gamma \Big(w_{ij}^2 + 2 \: w_{ij} \: \mathbb{E}[L_i^{(j-1)}] \Big) + \frac{1}{2} \: \phi_{ij}^2 \: w_{ij}^2 \: x_{ij} + \phi_{ij} \: w_{ij} \:\vb^{(j-1)}(i)$
\item $\vb^{(j)}(i) \geq \vb^{(j-1)}(i) + w_{ij} \: x_{ij} \: \phi_{ij}$
\item $\alpha_{ij} =\min \left\{q_{ij}, \sqrt{2}\right\}$
\end{itemize}
respectively shown in \eqref{eq_potentialf}, \eqref{v_hat} and \eqref{eq_alphaij}. Replacing $f_{ij}$ and $\vb^{(j)}(i)$ gives us that \eqref{eq_SDP2_restate} is implied by:
\begin{align*}
\gamma \Big(w_{ij}^2 &+ 2 \: w_{ij} \: \mathbb{E}[L_i^{(j-1)}] \Big) + \frac{1}{2} \: \phi_{ij}^2 \: w_{ij}^2 \: x_{ij} + \phi_{ij} \: w_{ij} \: \vb^{(j-1)}(i) \\
& \leq \left(1 - \frac{\alpha_{ij}^2}{2}\right) w_{ij}^2 + \alpha_{ij} \: w_{ij} \: \vb^{(j-1)}(i) + \alpha_{ij} \: w_{ij}^2 \: x_{ij} \: \phi_{ij}.
\end{align*}
Let us now use Lemma \ref{lem:lower_bound_v0} -- implying that $\vb^{(j-1)}(i) \geq \left(\beta + \maclambda \right)\mathbb{E}[L^{(j-1)}_i]$ -- for the left-hand side, divide both sides by $w_{ij}^2$ and use the definition of $q_{ij} = \vb^{(j-1)}(i)/w_{ij}$ to get that the above equation is implied by:

\begin{equation}
\label{eq_dual_feas_expanded}
\gamma + \frac{2 \gamma}{\beta + \maclambda} q_{ij} + \frac{1}{2} \: \phi_{ij}^2 \: x_{ij} + \phi_{ij} q_{ij} \leq 1 - \frac{\alpha_{ij}^2}{2} + \alpha_{ij} \: q_{ij} + \alpha_{ij} \: x_{ij} \: \phi_{ij}.
\end{equation}

\begin{figure}
    \centering
    \includegraphics[width = 0.6\textwidth]{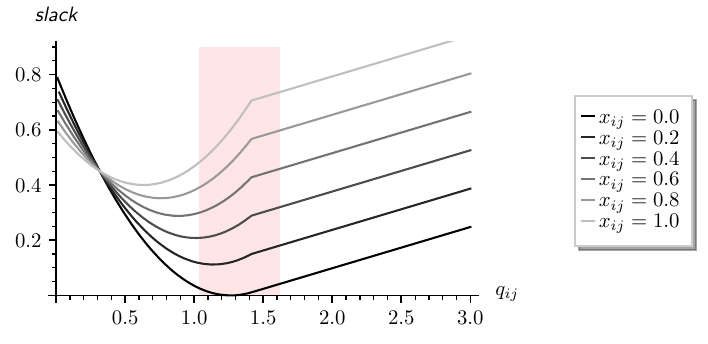}
    \caption{The \emph{slack} of \eqref{eq_dual_feas_expanded} for different values of $x_{ij}$ under the assumption that $\delta = \maclambda = 0$ and $\gamma = 1/5$. Observe that the inequality is tight for $q_{ij}=2\sqrt{2/5}\approx 1.265$ and $x_{ij}=0$. The plot also shows the interval $[a,b]$ in red.\label{fig:plot_slack}}
\end{figure}

\begin{remark}
For intuition, Figure \ref{fig:plot_slack} plots the \emph{slack} of this inequality \eqref{eq_dual_feas_expanded} (i.e. the right-hand side minus the left-hand side) for different values of $x_{ij}$ and $q_{ij}$ under the assumption that $\delta = \maclambda = 0$, and where $\alpha_{ij} = \min\{q_{ij}, \sqrt{2}\}, \beta = \sqrt{2/5}, \gamma = 1/5$.
Clearly, in order to get a competitive ratio better than $5$, one needs to slightly increase $\gamma$ while keeping this constraint satisfied. However, one can see that this inequality is tight when $q_{ij}=2\sqrt{2/5}\approx 1.265$ and $x_{ij}=0$, see Figure \ref{fig:plot_slack}. Increasing $\gamma$ could thus result in the above inequality being violated by a hard job satisfying $q_{ij}\in [a,b]$ and $x_{ij}\in [0,\theta]$, since $2\sqrt{2/5}\in [a,b]$. This also showcases the importance of Lemma \ref{lem:lower_bound_v0} and of the $\maclambda$ constant in \eqref{eq_dual_feas_expanded}, which is the key that allows to argue dual feasibility for hard jobs as well for a slightly increased $\gamma$.
\end{remark}

For simplicity of notation, since we always fix a job $j \in \J$ and a machine $i \in \M$ in the remainder of the proof, let us denote $\x := x_{ij},\q := q_{ij}$ and

\begin{align}
\label{eq_phi_simplified}
\phi:= \begin{cases} \beta + \delta &\text{if $j$ is easy for $i$,} \\
\beta  &\text{if $j$ is hard for $i$.}
\end{cases}
\end{align}

\paragraph*{Case 1.} In the case that $\alpha_{ij} = q_{ij}$, after rearranging terms, \eqref{eq_dual_feas_expanded} becomes: 
\begin{equation}
\label{eq_dual_case1}
g_1(\x, \q) \leq 0
\end{equation}
where the function $g_1 :[0,1] \times [0,\sqrt{2}] \to \mathbb{R}$ is defined as:
\begin{equation}
\label{eq_def_g1}
g_1(\x, \q) := \gamma -1 + \frac{1}{2}\phi^2 \x - \phi \: \x \: \q + \left(\phi + \frac{2\gamma}{\beta + \maclambda}\right) \q - \frac{1}{2} \: \q^2.
\end{equation}

\paragraph*{Case 2.} In the case that $\alpha_{ij} = \sqrt{2}$, after rearranging terms, \eqref{eq_dual_feas_expanded} becomes:
\begin{equation}
\label{eq_dual_case2}
g_2(\x, \q) \leq 0
\end{equation}
where the function $g_2 :[0,1] \times [\sqrt{2},\infty) \to \mathbb{R}$ is defined as:
\begin{equation}
\label{eq_def_g2}
g_2(\x, \q):= \gamma + \left(\frac{1}{2}\phi^2 - \sqrt{2} \: \phi \right) \x + \left(\frac{2 \gamma}{\beta + \maclambda} + \phi -\sqrt{2} \right) \q.
\end{equation}

We now have four cases to check on a partition of $[0,1] \times [0, \infty)$, see Figure \ref{fig_regions}. Regions $R_1$ and $R_2$ correspond to hard jobs, whereas $R_3$ and $R_4$ correspond to easy jobs.
\begin{itemize}
\item $R_1 := [0,\theta] \times [a, \sqrt{2}]$ 
\item $R_2 := [0,\theta] \times [\sqrt{2}, b)$
\item $R_3 := \Big([0, \theta] \times [0,a]\Big) \cup \Big([\theta, 1] \times [0, \sqrt{2}]\Big)$
\item $R_4:= \Big([0, \theta] \times [b,\infty)\Big) \cup \Big([\theta, 1] \times [\sqrt{2}, \infty)\Big)$.
\end{itemize}
\begin{figure}
\center
\includegraphics[width = 0.45\textwidth]{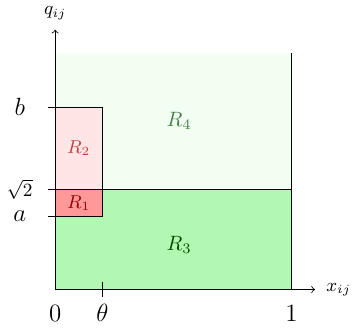}
\caption{The four different cases where dual feasibility needs to be checked. The two green regions correspond to easy jobs, whereas the two red regions correspond to hard jobs. The horizontal line at $\sqrt{2}$ differentiates between the cases where $q_{ij} < \sqrt{2}$, implying $\alpha_{ij} = q_{ij}$, and the case $q_{ij} \geq \sqrt{2}$, which implies $\alpha_{ij} = \sqrt{2}$.}
\label{fig_regions}
\end{figure}
In order to show dual feasibility, it is now enough to prove the following lemma.
\begin{lemma}
\label{lemma_dual_feas}
The function $g_1 : [0,1] \times [0,\sqrt{2}]$ satisfies:
\begin{align*}
g_1(\x,\q) \Bigr|_{\phi = \beta } \leq 0 \qquad \forall (\x, \q) \in R_1 ; \qquad g_1(\x,\q) \Bigr|_{\phi = \beta + \delta} \leq 0 \qquad \forall (\x, \q) \in R_3.
\end{align*}
Moreover, the function $g_2 : [0,1] \times [\sqrt{2}, \infty)$ satisfies:
\[g_2(\x,\q) \Bigr|_{\phi = \beta} \leq 0 \qquad \forall (\x, \q) \in R_2; \qquad g_2(\x,\q) \Bigr|_{\phi = \beta + \delta} \leq 0 \qquad \forall (\x, \q) \in R_4.\]
This implies that the constructed dual solution is feasible.
\end{lemma}
\begin{proof}
The function $g_2$ is linear in both $x$ and $q$, so we know that its supremum is attained at a point that is on the boundary for both $x$ and $q$. Hence, it suffices to check that $g_2(\x,\q) |_{\phi = \beta } \leq 0$ for \[(\x,\q)\in \Big\{(0,\sqrt{2}),(\theta,\sqrt{2}),(\theta, b), (0, b)\Big\}\] and $g_2(\x,\q) |_{\phi = \beta+\delta} \leq 0$ for 
\[(\x,\q)\in \Big\{(0,b),(\theta ,\sqrt{2}),(1,\sqrt{2})\Big\} \text{ and } \q \to \infty \text{ with } \x=0 \text{ and } \x=1.\]
The function $g_1$ is linear in $x$, so any supremum will be attained at a boundary point for $x$. For $q$, we need to check boundary points, in addition to points that satisfy $\frac{d}{d q} \: g_1(\x,\q)=0$, meaning that $\q=\hat{q}(\x, \phi):=\phi(1-x)+2\gamma/(\beta+\maclambda)$ after a small computation. Hence, it suffices to check  that 
$g_1(\x,\q) |_{\phi = \beta } \leq 0$ for 
\[(\x,\q)\in \Big\{(0,a),(0,\sqrt{2}),(\theta, a), (\theta, \sqrt{2}), (0, \hat{q}(0, \beta)), (\theta, \hat{q}(\theta, \beta))\Big\}.\]
Finally, we need to check that $g_1(\x,\q) |_{\phi = \beta+\delta} \leq 0$ for 
\[(\x,\q)\in \Big\{(0,0),(0,a),(\theta,\sqrt{2}),(1, 0), (1, \sqrt{2}), (1, \hat{q}(1, \beta +\delta))
\Big\}\] and that $(0, \hat{q}(0, \beta+\delta))\notin R_3$. It can be checked that all of these constraints are satisfied by our choice of constants for $a, b, \theta, \gamma, \beta,\delta, \maclambda$ in Section \ref{sec_constants}.
\end{proof}

\section{Lower bounds}
\label{sec_lower_bounds}

\subsection{A matching lower bound for fractional algorithms}
\label{sec:lower_bound_frac}

In this section, we show that our fractional algorithm described in \cref{opt_fractional_algo} is optimal by providing a matching lower bound against any fractional algorithm.
\begin{theorem}
\label{thm_hardness_fractional}
For any $\varepsilon > 0$, there exists an online instance to the online unrelated load balancing problem such that any fractional algorithm has a competitive ratio of at least $4 - \varepsilon$.
\end{theorem}

\begin{proof}
Let $n \in \mathbb{N}$, the instance will satisfy $\M = \J = [n]$. Consider a uniformly random permutation $\sigma: [n] \to [n]$ of the machines. The feasible machines for every job $j$ will now be:
\[\mathcal{S}_j = \{\sigma(i) :  j \leq i \leq n \} \qquad \forall j \in [n].\]
Consider the function $f : [0,1] \to \mathbb{R}$ defined as 
$f(x) = 1/\sqrt{1-x}.$
The weights of the instance are then defined as \[w_{ej} = w_j = f\left(\frac{j-1}{n}\right) \qquad \forall j \in \J, \forall e \in \mathcal{S}_j.\] Observe that the weights do not depend on the machine, meaning that the lower bound will hold even in the restricted identical machines setting. To simplify computations, we will below often use the following lower and upper approximations of a sum by an integral for an arbitrary increasing function $g:[0,1] \to \mathbb{R}$ and any $k \leq n$:
\begin{align}
\label{eq_integral_approx}
\int_0^{(k-1)/n} g(x) dx \leq \frac{1}{n}\: \sum_{i = 1}^k g\left(\frac{i-1}{n}\right) \leq \frac{1}{n} \: g\left(\frac{k-1}{n}\right) + \int_0^{(k-1)/n} g(x) dx.
\end{align}

\begin{figure}
	\begin{subfigure}{0.57\textwidth}
	\raisebox{0.7cm}{\includegraphics[width = \textwidth]{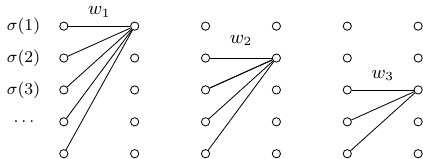}}
	\end{subfigure}
	\hspace{0.6cm}
	\begin{subfigure}{0.35\textwidth}
	\center
	\includegraphics[width = \textwidth]{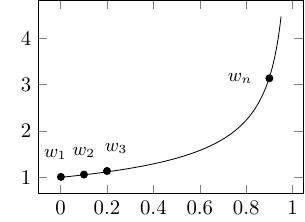}
	\end{subfigure}
	\caption{Illustration of the adversarial instance and its weight distribution (for $n = 10$) against fractional and independent rounding algorithms.}
	\label{fig_adversarial_instance}
\end{figure}

Consider the solution $x^*$ where each job $j \in \J$ is assigned to $\sigma(j) \in \M$. This is in fact the optimal offline solution with value:
\begin{align}
\label{eq_opt_fractional}
\sum_{i=1}^n L_i(x^*)^2 &= \sum_{i=1}^n f\left(\frac{i-1}{n}\right)^2 \leq f\left(\frac{n-1}{n}\right)^2 + n\int_{0}^{(n-1)/n} f(x)^2 \: dx \nonumber\\&
= n + n\int_{0}^{(n-1)/n} \frac{1}{1-x} dx = n + n\left(\log\left(1\right)-\log\left(\frac{1}{n}\right)\right)= n \: (\log(n)+1).
\end{align}

Let us now fix an arbitrary deterministic online fractional algorithm $\mathcal{A}$ generating a solution that we will denote by $x$. At the arrival of $j \in \J$, due to the random permutation, we have that $\mathbb{E}[x_{\sigma(i),j}] = \mathbb{E}[x_{\sigma(i'),j}]$ for every $i,i' \geq j$. Since the fractional algorithm exactly sends a total fractional value of one, we get that:
\[\mathbb{E}[x_{\sigma(i),j}] = \frac{1}{n-j+1} \qquad \forall i \in \{j, \dots, n\}.\]
We therefore get that for every $i \in \M$:
\begin{align*}
\mathbb{E}\Big[L_{\sigma(i)}(x)\Big]&= \sum_{j=1}^i w_{\sigma(i),j} \: \mathbb{E}[x_{\sigma(i),j}] = \sum_{j=1}^i f\left(\frac{j-1}{n}\right)\frac{1}{n - j + 1} \geq \int_{0}^{(i-1)/n} \frac{ f(x)}{1-x} dx \\ &=  \int_{0}^{(i-1)/n} (1-x)^{-3/2} \: dx = \frac{2}{\sqrt{1-x}} \: \bigg\vert ^ {(i-1)/n}_0 = 2 \: f\left(\frac{i-1}{n}\right)-2.
\end{align*}
By using Jensen's inequality, we can now lower bound the value obtained by the algorithm:
\begin{align}
\label{eq_cost_fractional}
\sum_{i=1}^n \mathbb{E}\left[L_{\sigma(i)}(x) ^2\right] &\geq \sum_{i=1}^n \mathbb{E}\left[L_{\sigma(i)}(x)\right]^2 \geq \sum_{i = 1}^n \left( 2 \: f\left(\frac{i-1}{n}\right)-2 \right)^2 \nonumber\\
&\geq 4\sum_{i=1}^n\left(f\left(\frac{i-1}{n}\right)^2-2 \: f\left(\frac{i-1}{n}\right)\right) \nonumber\\
        &\geq 4 \: n \: \int_{0}^{(n-1)/n}\Big(f(x)^2-2f(x)\Big)dx = 4n\log(n) - 16n\left(1-\sqrt{\frac{1}{n}}\right).
\end{align}
We thus easily see that the competitive ratio tends to $4$ from below when tending $n$ to infinity. For a fixed $\varepsilon > 0$, picking $n$ to be large enough finishes the proof.
\end{proof}

\subsection{A matching lower bound for independent rounding algorithms}
\label{sec_lower_bound_ind_rounding}
In this section, we show that any randomized algorithm making \emph{independent} random choices for each job $j \in \J$ cannot be better than $5$-competitive. This shows that Algorithm \ref{caragiannis_algo} is optimal for this class of randomized algorithms.

\begin{theorem}
For any $\varepsilon > 0$, there exists an instance to the online unrelated load balancing problem such that any randomized algorithm making independent random choices for every job has a competitive ratio of at least $5 - \varepsilon$.
\end{theorem}
\begin{proof}
Let us fix such a randomized algorithm $\mathcal{A}$. The instance is exactly the same as in the proof of Theorem \ref{thm_hardness_fractional} and is illustrated in Figure \ref{fig_adversarial_instance}. Let us denote by $X_{ij} \in \{0,1\}$ the indicator random variable of whether $j$ is assigned to $i$. Note that we now have two sources of randomness: both the random permutation and the random choices of the algorithm. The cost of a randomized algorithm can be written as follows:
\[\sum_{i \in \M} \mathbb{E}\left[L_i^2\right] = \sum_{i \in \M} \mathbb{E}\left[L_i\right]^2 + \sum_{i \in \M} \text{Var} [L_i]\]
where the expectation is both over the random permutation and the random choices of the algorithm. Now note that, for a fixed permutation $\sigma$, we can interpret $x_{\sigma(i),j}:= \underset{\mathcal{A}}{\mathbb{E}}[X_{\sigma(i),j}]$ as a fractional algorithm, where the expectation is only over the random choices of $\mathcal{A}$.
By \eqref{eq_cost_fractional}, we then have that
\begin{align}
\label{eq_fractional_part}
\sum_{i \in \M} \mathbb{E}\left[L_i\right]^2 = \sum_{i \in \M} \mathbb{E}\left[L_{\sigma(i)}\right]^2 \geq 4n\log(n) - 16n\left(1-\sqrt{\frac{1}{n}}\right).
\end{align}
For the second term, note that 
\begin{align*}
\sum_{i \in \M}\text{Var} [L_i] &= \sum_{i \in \M} \sum_{j,k \in \J} w_{j} \: w_{k} \Big(\mathbb{E}[X_{ij}X_{ik}] - \mathbb{E}[X_{ij}] \mathbb{E}[X_{ik}]\Big) \\
&= \sum_{i \in \M}\sum_{j \in \J} w_{j}^2 \: \Big(\mathbb{E}[X_{ij}] - \mathbb{E}[X_{ij}]^2\Big).
\end{align*}
where the last equality uses $\mathbb{E}[X_{ij} X_{ik}] = \mathbb{E}[X_{ij}] \mathbb{E}[X_{ik}]$ for $j \neq k$, due to our independence assumption. Let us first compute a lower bound for the first term:
\[\sum_{i \in \M} \sum_{j \in \J} w_{j}^2 \: \mathbb{E}[X_{ij}] = \sum_{j \in \J} w_j^2 = \sum_{j =1}^n f\left(\frac{j-1}{n}\right)^2 \geq n \int_{0}^{(n-1)/n} f(x)^2 dx = n \log(n)\]
where we use $\sum_{i \in \M} \mathbb{E}[X_{ij}] = 1$ for the first equality and the approximation \eqref{eq_integral_approx} for the inequality. Note that at the arrival of $j \in \J$, due to the random permutation, we have:
\[\mathbb{E}[X_{\sigma(i),j}] = \underset{\sigma}{\mathbb{E}} \: [x_{\sigma(i),j}] = \frac{1}{n-j+1} \qquad \forall i \in \{j, \dots, n\}.\]

\noindent Therefore, we can upper bound the second term as follows:
\begin{align*}
\sum_{i \in \M} \sum_{j \in \J} w_{j}^2 \: \mathbb{E}[X_{ij}]^2 &= \sum_{j \in \J} w_j^2 \sum_{i \in \M} \mathbb{E}[X_{\sigma(i),j}]^2= n \sum_{j = 1}^n \sum_{i = j}^n \frac{1}{(n-j+1)^3} \\ &= n \sum_{j = 1}^n \frac{1}{(n-j+1)^2} = n \sum_{j = 1}^n \frac{1}{j^2} < \frac{\pi^2}{6} n
\end{align*}
where the last inequality uses that $\sum_{j = 1}^\infty 1/j^2 = \pi^2/6$. We thus have that
\begin{equation}
\label{eq_variance}
\sum_{i \in \M} \text{Var} [L_i] > n \log(n) - \frac{\pi^2}{6} n.
\end{equation}
We have seen that the optimal solution has cost $n (\log(n)+1)$ in \eqref{eq_opt_fractional}. Hence, by \eqref{eq_fractional_part} and \eqref{eq_variance}, we see that the competitive ratio tends to $5$ as $n$ tends to infinity.
\end{proof}

\subsection{A matching lower bound for \texorpdfstring{$R || \sum w_j C_j$}{R || sum wj Cj}}
In this section, we prove a matching lower bound against fractional algorithms for the online scheduling problem $R || \sum w_j C_j$, defined as follows. A set of machines $\M$ is given and a set of jobs $\J$ arrives online in an adversarial order. Each time a job $j \in \J$ arrives, it reveals a weight $w_j \geq 0$, and a subset $\mathcal{S}_j \subset \M$ of machines it can be assigned to with unrelated processing times $p_{ij} \geq 0$ for every $i \in \mathcal{S}_j$, at which point an online algorithm needs to assign this job to a machine. Once every job has been assigned, every machine reorders the jobs assigned to it by increasing Smith ratios $\delta_{ij} := p_{ij}/w_j$, which is known to be the optimal offline ordering for this problem. We denote the order induced by this rule on each machine by $k \prec_i j$ (whenever $\delta_{ik} \leq \delta_{ij}$), where ties are broken arbitrarily. The completion time of every job is defined as:
\[C_j(x) = \sum_{i \in \mathcal{S}_j} x_{ij} \Big(p_{ij} + \sum_{k \prec_i j} p_{ik} \: x_{ik}\Big). \]
The goal of the problem is to minimize the sum of weighted completion times:
\[C^{(SR)}(x) := \sum_{j \in \J} w_j \: C_j(x).\]

It is known that a simple online greedy algorithm achieves a competitive ratio of $4$ for this problem \cite{gupta2020greed}. We show here that the greedy algorithm is optimal, even among fractional or randomized algorithms. To do so, we use the hard instance constructed for the unrelated load balancing problem in Theorem \ref{thm_hardness_fractional}. We now describe why this model is related to the load balancing problem in the restricted identical machines setting, meaning that $w_{ij} \in \{w_j, \infty\}$. Given an instance $\mathcal{I}$ and a fractional solution $x \in [0,1]^{M \times \J}$, the square of the $\ell_2$ norm of the loads is:
\[C^{(LB)}(x) := \sum_{i \in M} \left(\sum_{j \in \J} w_{j} x_{ij}\right)^2.\]
Consider now the same instance in the second model under uniform Smith ratios, meaning that the processing times are set to $p_{ij} = w_j$ for every $i \in \mathcal{S}_j$. Given a fractional solution $x$, the cost in this model can be computed to be:
\begin{align}
\label{cost_smith_rule}
C^{(SR)}(x) = \frac{1}{2} \: C^{(LB)}(x) + \sum_{i \in \M} \sum_{j \in \J} w_j^2 \left(x_{ij} - \frac{1}{2} x_{ij}^2\right).
\end{align}
The idea of the proof will now be to modify the instance in Theorem \ref{thm_hardness_fractional} in order to make the second term above negligible. 

Let $\mathcal{I}$ be the instance constructed in Theorem \ref{thm_hardness_fractional} and illustrated in Figure \ref{fig_adversarial_instance}. We now consider a modified instance $\mathcal{I}(t)$, which for each job $j \in \J$, has $t \in \mathbb{N}$ copies of the same job arriving consecutively, all of which have the same feasible machines. These copies of a job $j$ are denoted as $K(j)$. Moreover, for every $k \in K(j)$, we set the weight to be $w_k = w_j / t$. The new set of jobs is denoted by $\tilde{\J} = \cup_{j \in \J} K(j)$. 

\begin{lemma}
\label{lemma_cost_correspondence}
For a fractional solution $y \in [0,1]^{M \times \tilde{\J}}$ to $\mathcal{I}(t)$ and a fractional solution $x \in [0,1]^{M \times \J}$ to $\mathcal{I}$ satisfying $\sum_{k \in K(j)} y_{ik} = t \: x_{ij}$ for every $j \in \J$ and every $i \in M$, we have
\[\frac{1}{2} \: C^{(LB)}(x) \leq C^{(SR)}(y) \leq \frac{1}{2} \: C^{(LB)}(x) + \frac{1}{t} \sum_{j \in \J} w_j^2.\]
\end{lemma}

{
\allowdisplaybreaks
\begin{proof}
By \eqref{cost_smith_rule}, we have that
\begin{align*}
C^{(SR)}(y) &= \frac{1}{2}\sum_{i \in M} \left(\sum_{j \in \J} \sum_{k \in K(j)} w_{k} y_{ik}\right)^2 + \sum_{i \in M} \sum_{j \in \J} \sum_{k \in K(j)} w_k^2 \left(y_{ik} - \frac{1}{2} y_{ik}^2\right) \\
&= \frac{1}{2}\sum_{i \in M} \left(\sum_{j \in \J}\frac{w_j}{t} \sum_{k \in K(j)} y_{ik}\right)^2 + \sum_{i \in M} \sum_{j \in \J} \frac{w_j^2}{t^2}\sum_{k \in K(j)} \left(y_{ik} - \frac{1}{2} y_{ik}^2\right) \\
&= \frac{1}{2}\sum_{i \in M} \left(\sum_{j \in \J}w_j x_{ij}\right)^2 + \frac{1}{t} \sum_{j \in \J} w_j^2 - \frac{1}{2t^2} \sum_{i \in M} \sum_{j \in \J} w_j^2 \sum_{k \in K(j)} y_{ik}^2 \\
&\leq \frac{1}{2} \: C^{(LB)}(x) + \frac{1}{t} \sum_{j \in \J} w_j^2.
\end{align*}
By looking at the first line of the above computation, and observing that $y_{ik} \geq y_{ik}^2 \geq \frac{1}{2} y_{ik}^2$, since $y_{ik} \in [0,1]$, we also get $C^{(SR)}(y) \geq \frac{1}{2} \: C^{(LB)}(x)$.
\end{proof}
}

\begin{theorem}
\label{thm_hardness_fractional_SR}
For any $\varepsilon > 0$, there exists an online instance to the online scheduling problem $R || \sum w_j C_j$ such that any fractional algorithm has a competitive ratio of at least $4 - \varepsilon$.
\end{theorem}
\begin{proof}
Note that we can easily convert a solution $x$ of $\mathcal{I}$ into a solution of $\mathcal{I}(t)$ by setting $y_{ik} = x_{ij}$ for every job $k \in K(j)$. By Lemma \ref{lemma_cost_correspondence}, this shows that the optimal solution $y^*$ of instance $\mathcal{I}(t)$ has cost at most
\begin{equation}
\label{eq_opt_sr}
C^{(SR)}(y^*) \leq \frac{1}{2} \: C^{(LB)}(x^*) + O\left(\frac{1}{t}\right) = \frac{1}{2} \: n \: (\log(n)+1) + O\left(\frac{1}{t}\right)
\end{equation}
where $x^*$ is the optimal solution on instance $\mathcal{I}$ whose cost we computed in \eqref{eq_opt_fractional}. Conversely, we can convert a fractional solution $y$ obtained by an arbitrary online algorithm on $\mathcal{I}(t)$ to a solution generated online on $\mathcal{I}$ by setting $x_{ij} = \sum_{k \in K(j)} y_{ik}/t$. By \eqref{eq_cost_fractional}, the cost of that solution on $\mathcal{I}$ is at least:
\[
C^{(LB)}(x) \geq 4n\log(n) - 16n\left(1-\sqrt{\frac{1}{n}}\right).
\]
By Lemma \ref{lemma_cost_correspondence}, we have that any online fractional algorithm incurs cost at least 
\begin{equation}
\label{eq_lower_bound_sr}
C^{(SR)}(y) \geq \frac{1}{2} \left(4n\log(n) - 16n\left(1-\sqrt{\frac{1}{n}}\right) \right)
\end{equation}
Hence, due to \eqref{eq_opt_sr} and \eqref{eq_lower_bound_sr}, and by letting $n$ and $t$ tend to infinity, we get the proof of the theorem.
\end{proof}

Even though the objective function for this problem is non-convex (in contrast to the load balancing problem), we can also adapt this lower bound to randomized algorithms. The expected cost of a randomized algorithm generating a random assignment $X$ in this model under uniform Smith ratios is
\[\mathbb{E}[C^{(SR)}(X)] = \frac{1}{2} \left(\sum_{i \in M}\mathbb{E}[L_i(X)^2] + \sum_{j \in J} w_j^2 \right).\]
By convexity, we have $L_i(x)^2 \leq \mathbb{E}[L_i(X)^2]$. Hence, by looking at \eqref{cost_smith_rule}, we get that the cost of the induced fractional algorithm is upper bounded by
\[C^{(SR)}(x) \leq \mathbb{E}[C^{(SR)}(X)] + \frac{1}{2} \sum_{i \in M} \sum_{j \in J}w_j^2 x_{ij}(1-x_{ij}).\]
The second term becomes negligible compared to $t$ in the construction $\mathcal{I}(t)$ by the same computation as in \cref{lemma_cost_correspondence}. By \cref{thm_hardness_fractional_SR}, this shows that no randomized algorithm can achieve a competitive ratio than $4$.

\section{Concluding remarks}
We presented the first randomized algorithm breaking the barrier of $5$ for the online load balancing on unrelated machines under the objective of minimizing the (square of the) $\ell_2$ norm of the loads of the machines. The two main ingredients were a new primal-dual approach based on a natural semidefinite programming relaxation of the problem and an online implementation of a known correlated randomized rounding algorithm. We also gave new simple and unified analyses of the previously best known deterministic and randomized algorithms, designed an optimal $4$-competitive fractional algorithm and provided several matching lower bounds. 

We believe that the competitive ratio achieved of $\approx 4.98$ can be further improved. A possible direction for improvement is to design a more clever grouping algorithm that forms different groups based on job weights. Another online correlated rounding algorithm and a different dual fitting approach on this semidefinite program could also be promising directions. Other possible future work includes using this SDP dual fitting approach for other online problems with a quadratic objective function, or trying to design a similar primal-dual approach for the $\ell_p$ norm objective with $p >2$.

\section*{Use of AI tools}
This paper was first posted on ArXiv in November 2025 and AI was not used in finding the results. AI helped with minor changes to the text and language checking.

\bibliographystyle{alpha}
\bibliography{references}

\appendix

\section{Details of the online correlated rounding}
\label{sec_online_dep_rounding}
We describe here how we can adapt \cref{algo_offline_rounding} for our online setting. For every machine $i \in \M$, every $G\in G_i$ and every $\ell \in \mathbb{N}$, we assume that our algorithm has randomly sampled a variable $R_{G,\ell} \in [0,1]$ uniformly at random from the interval $[0,1]$ as soon as this group is created online. 
{\renewcommand{\baselinestretch}{1.3}\selectfont{
\begin{algorithm}
    \caption{Online dependent rounding}
    \label{algo_online_rounding}
    \begin{algorithmic}
    \State \textbf{Input:} The fractional assignment $x$ and the grouping $\mathcal{G}_i$ for each $i \in \M$
    \When{$j \in \J$ arrives}
    \State $\ell : = 1$
    \While{$j$ is unassigned}
    \For{$i \in M$}
    \State Let $G = G_i(j)$ be the set in $\mathcal{G}_i$ containing $j$
    \If{$R_{G,\ell}\in [0,x_{ij})$}\State  Sample $N_{ij} \sim \widetilde{\operatorname{Poi}}(x_{ij})$ \Else\State Set $N_{ij}=0$\EndIf
    \State Update $R_{G,\ell} = R_{G,\ell}-x_{ij}$
    \EndFor
    \If{$S:= \sum_{i \in \mathcal{S}_j} N_{ij} > 0$} 
    \State Assign $j$ to $i$ (i.e. set $X_{ij} =1$) with probability $N_{ij}/S$
    
    \ElsIf{$S = 0$} 
    \State Update $\ell = \ell+1$
    \EndIf
    \EndWhile
    \State \Return $(X_{ij})_{i \in \mathcal{S}_j}$
    \EndWhen
    \end{algorithmic}
 \end{algorithm}
}}
\begin{lemma}
After the online arrival of all the jobs, Algorithms \ref{algo_offline_rounding} and \ref{algo_online_rounding} are equivalent.
\end{lemma}
\begin{proof}
Let us fix a machine $i \in M$, a group $G \in \mathcal{G}_i$ and a job $j \in G$. The only thing we need to show is that this group recommends $j$ with probability $x_{ij}$ for each iteration $\ell > 0$, which happens if and only if the first \emph{if} statement in Algorithm \ref{algo_online_rounding} is executed. Let us denote by $R_{G,\ell} \in [0,1]$ the initial random value sampled by this algorithm for this group and this iteration.

Let us denote by $B_{ij}(\ell) \in \{0,1\}$ the random variable indicating whether the first if statement in Algorithm \ref{algo_online_rounding} is executed at iteration $\ell > 0$. A necessary condition for this to occur is that the event $\mathcal{E} := \{R_{G,\ell} \geq \sum_{k \in G : k < j} x_{ik}\}$ happens, since otherwise the updated value $R_{G,\ell}$ when $j$ arrives would be negative. Therefore, we have
\begin{align*}
\mathbb{P}[B_{ij}(\ell) = 1] &= \mathbb{P}[B_{ij}(\ell) = 1, \mathcal{E}] = \mathbb{P}[B_{ij}(\ell) = 1 \mid \mathcal{E}] \: \mathbb{P}[\mathcal{E}] \\&= \frac{x_{ij}} {1 - \sum_{k \in G : k < j} x_{ik}} \left(1 - \sum_{k \in G : k < j} x_{ik}\right) = x_{ij}. \qedhere
\end{align*}
\end{proof}

\section{Hypergraph/congestion generalization}
\label{sec_hypergraph_gen}
A generalization of this problem is when each $j \in \J$ now has a collection $\mathcal{S}_j \subseteq 2^\M$ of \emph{subsets} or \emph{hyperedges} of machines it can be assigned to. In that case, we denote by $x_{ij} \in \{0,1\}$ the indicator variable of whether $j$ is assigned to the hyperedge $i \in \mathcal{S}_j$ and $z_{ej} = \sum_{i \in S_j : e \in i} x_{ij}$ whether $j$ is assigned to a hyperedge containing $e \in \M$. The load of a machine $e \in M$ is still the total amount of weight assigned to it:
\[L_e(x) = \sum_{j \in \J} w_{ej} \: z_{ej} = \sum_{j \in \J} w_{ej} \sum_{i \in S_j : e \in i} x_{ij}.\]
The goal of the problem is again to minimize $\sum_{e \in \M} L_e(x)^2$. 

One motivation to study this problem could for instance be the following. Let $G = (V,E)$ be an undirected graph and consider the edges $E = M$ as resources (or machines). Suppose that at each arrival $j \in J$, an online request arrives to pick a path between two nodes $s_j \in V$ and $t_j \in V$, meaning that $\mathcal{S}_j$ is the set of all paths in the graph $G$ between $s_j$ and $t_j$. An online algorithm then irrevocably needs to pick some path $i \in \mathcal{S}_j$ to serve this request. This model has been thoroughly studied from a game-theoretic perspective and is called a linear congestion game in the algorithmic game theory literature. In this model, the dual SDP becomes:
\begin{align*}
    \max \sum_{j \in \J} y_j  - \frac{1}{2}  &\Vert \vb \Vert ^ 2 \\
    y_j &\leq \sum_{e \in i}  w_{ej}^2 - \frac{1}{2}\Vert v_{ij} \Vert ^ 2 + \langle \vb, v_{ij} \rangle \qquad \forall j \in \J, \forall i \in \mathcal{S}_j\nonumber\\
    \langle v_{ij}, v_{i'k} \rangle &\leq 2 \sum_{e \in i \cap i'} w_{ej} \: w_{ek} \hspace{2.6cm} \forall (i,j) \neq (i',k) \text{ with } j,k \in \J. \nonumber
\end{align*}

The analysis of \Call{Greedy}{} extends to this model in a straightforward way, and we briefly highlight the differences. The equilibrium inequalities become the following:

\[C^{(j)}(x) - C^{(j-1)}(x) \leq \sum_{e \in i} \left( w_{ej}^2 + 2 \: L_e^{(j-1)}(x)\: w_{ej} \right) \qquad \forall j \in \J, \forall i \in \mathcal{S}_j, \]
and the dual fitting allowing to analyze the algorithm is the same as in the standard model, except that $v_{ij}(e) := \alpha \: w_{ej} \: \mathds{1}_{\{e \in i\}}$.

\section{Computation of the dual SDP}
\subsection{Taking the dual}
\label{sec_take_dual}
The same derivation is done in \cite{kashaev2026selfish} and we only include it here for completeness. For two matrices $A, B \in \mathbb{R}^{n \times n}$, the \emph{trace inner product} is defined as:
\[\langle A, B \rangle := \text{Tr}(A^TB) = \sum_{i,j = 1}^n A_{ij} \: B_{ij}.\]
The objective function \eqref{eq_primal_sdp} of our primal SDP can then be written as $\langle C, X \rangle$, where the symmetric matrix $C$ is of dimension $1 + \sum_{j \in \J} |\mathcal{S}_j|$ and is defined as $C_{\{0,0\}} = 0, C_{\{0, ij\}} = 0$, 
\[C_{\{ij, \: ij\}} = w_{ij}^2 \quad , \quad C_{\{ij, \: i'k\}} = w_{ij} \: w_{i'k} \: \mathds{1}_{\{i = i'\}}.\] In the hypergraph model, this matrix becomes $C_{\{0,0\}} = 0, C_{\{0, ij\}} = 0$, 
\[C_{\{ij, \: ij\}} = \sum_{e \in i}w_{ej}^2 \quad, \quad C_{\{ij, \: i'k\}} = \sum_{e \in i \cap i'} w_{ej} \: w_{ek}.\] 
Our primal SDP relaxation is then the following.
\begin{align*}
    \min \langle C, X \rangle \qquad  \qquad &\\
    \sum_{i \in \mathcal{S}_j} X_{\{ij,\: ij\}} &= 1 \hspace{4cm} \forall j \in N \\
    X_{\{0,0\}} & = 1 \\
    X_{\{0, \: ij\}} &= X_{\{ij, \: ij\}} \hspace{3cm} \forall j \in N, i \in \mathcal{S}_j \\
    X_{\{ij, \: i'k\}} &\geq 0 \hspace{4cm} \forall {(i,j), (i',k)} \text{ with } j,k > 0. \\
    X & \succeq 0
\end{align*}
It can be easily checked that the following form of semidefinite programs is a primal-dual pair. The dual variables $(\lambda_i)_i$ and $(\mu_j)_j$ respectively correspond to the equality and inequality constraints, whereas the matrix variable $Y$ corresponds to the semidefinite constraint.

\vspace{0.3cm}
\noindent\begin{minipage}[t]{0.4\textwidth}
\begin{align*}
    \min \langle C, X \rangle & \\
    \langle A_i, X \rangle & = b_i \quad \forall i \\
    \langle B_j, X \rangle & \geq 0 \quad \: \forall j\\
    X &\succeq 0 
\end{align*}
\end{minipage}%
\begin{minipage}[t]{0.6\textwidth}
\begin{align*}
    \max \sum_{i} b_i &\lambda_i \\
    Y &= C - \sum_i \lambda_i A_i  - \sum_j \mu_j B_j \\
    Y &\succeq 0, \quad \mu \geq 0.
\end{align*}
\end{minipage}\par
\vspace{0.3cm}

\noindent Observe that our above primal SDP is in fact of that form. Let us denote by $(y_j)_{j \in N}, z$ and $(\sigma_{ij})_{j \in N, i \in \mathcal{S}_j}$ the dual variables respectively corresponding to the three sets of equality constraints. Let us denote by $\mu_{\{ij, i'k\}} \geq 0$ the dual variables corresponding to the inequality (or non-negativity) constraints. The dual objective becomes $\sum_{j \in N} y_j + z$ and the dual matrix equality becomes:
\begin{align*}
Y_{\{0,0\}} &= -z \\
Y_{\{0, \: ij\}} &= \frac{\sigma_{ij}}{2} \hspace{5.9cm} \forall j \in N, i \in \mathcal{S}_j\\
Y_{\{ij, \: ij\}} &= C_{\{ij, \; ij\}} - y_j - \sigma_{ij} - \mu_{\{ij, \: ij\}} \hspace{2.1cm} \forall j \in N, i \in \mathcal{S}_j\\
Y_{\{ij, \: i'k\}} &= C_{\{ij, \: i'k\}} - \mu_{\{ij,\:  i'k\}} \hspace{3.5cm} \forall (i,j) \neq (i',k) \text{ with } j,k >0.
\end{align*}
Note that we can now eliminate the dual variables $z$ and $\sigma$ by the first two equalities. Moreover, we can eliminate the $\mu \geq 0$ variables by replacing the last two equalities by inequalities. Let us now do the change of variable $Y' = 2Y$ and let the vectors of the Cholesky decomposition of $Y'$ be $v_0$ and $(v_{ij})_{j \in N, i \in \mathcal{S}_j}$, meaning that $Y'_{a, \: b} = \langle v_{a}, v_{b} \rangle$ holds for all the entries of $Y'$. The dual SDP in vector form can then be rewritten as:
\begin{align*}
    \max \sum_{j \in N} y_j  - &\frac{1}{2}\Vert v_0 \Vert ^ 2 \\
    y_j &\leq C_{\{ij, \; ij\}} - \frac{1}{2}\Vert v_{ij} \Vert ^ 2 - \: \langle v_0, v_{ij} \rangle \qquad \qquad \forall j \in N, i \in \mathcal{S}_j\\
    \langle v_{ij}, v_{i'k} \rangle &\leq 2 \: C_{\{ij, \: i'k\}} \hspace{4.3cm} \forall (i,j) \neq (i',k) \text{ with } j,k >0
\end{align*}
By doing a change of variable $\vb := - v_0$, we get the program \eqref{dual_sdp}.

\appendix
\end{document}